\documentclass[default]{ifcolog} 

\usepackage{graphicx} 
\usepackage{amsmath, amssymb, amsfonts, amsthm} 
\usepackage{mathrsfs} 
\usepackage[title]{appendix}
\usepackage{xcolor}
\usepackage{textcomp}
\usepackage{manyfoot}
\usepackage{booktabs}
\usepackage{algorithm, algorithmicx, algpseudocode}
\usepackage{listings}
\usepackage{float}
\usepackage{multirow, multicol}
\usepackage{rotating}
\usepackage{hhline}
\usepackage{balance} 
\usepackage[all]{nowidow}
\usepackage{proof}
\usepackage{tikz} 
\usetikzlibrary{decorations.pathreplacing}
\usepackage[font=footnotesize]{caption}
\usepackage{tabularx}
\usepackage{natbib}

\theoremstyle{plain}
\newtheorem{theorem}{Theorem}
\newtheorem{proposition}[theorem]{Proposition}
\newtheorem{lemma}[theorem]{Lemma}
\newtheorem{corollary}[theorem]{Corollary}

\theoremstyle{definition}
\newtheorem{definition}[theorem]{Definition}

\newtheorem{example}[theorem]{Example}
\newtheorem{remark}[theorem]{Remark}

\newcommand{\model}{\mathcal{M}}
\newcommand{\amodel}{\mathbf{A}}
\newcommand{\tella}{\mathsf{tell}}
\newcommand{\showa}{\mathsf{show}}

\newcommand{\tup}[1]{\langle #1 \rangle}
\newcommand{\set}[1]{\{#1\}}
\newcommand{\ra}{\rightarrow}
\newcommand{\iffdef}{\Leftrightarrow}
\newcommand{\lra}{\leftrightarrow}

\newcommand{\prop}{\mathsf{Prop}}
\newcommand{\agt}{\mathsf{Agt}}
\newcommand{\obsym}{\mathsf{Obs}}
\newcommand{\obs}[2]{\textsf{o}_{#1}{#2}}
\newcommand{\obsp}{\obs{a}{p}}
\newcommand{\obsnp}{\obs{a}{\neg p}}
\newcommand{\obsbp}{\obs{b}{p}}
\newcommand{\obsbnp}{\obs{b}{\neg p}}
\newcommand{\lit}{\mathsf{Lit}}
\newcommand{\Ob}{\mathcal{O}}
\newcommand{\B}{\mathcal{B}}

\newcommand{\Bdiam}{\hat{\mathcal{B}}}

\newcommand{\ES}{\Ob^{\mathsf{S}}}
\newcommand{\SB}{\B^{\mathsf{S}}}

\newcommand{\Sim}{\mathsf{Sim}}
\newcommand{\Dis}{\mathsf{Dis}}

\newcommand{\pre}{\mathsf{pre}}
\newcommand{\post}{\mathsf{post}}

\newcommand{\axm}[1]{\mathtt{#1}}

\newcommand{\dol}{\boldsymbol{\mathsf{DLM}}}
\newcommand{\ol}{\boldsymbol{\mathsf{OL}}}
\newcommand{\tslogic}{\dol^{(+,-)}}
\newcommand{\dolaxiom}{\mathsf{DOL}}
\newcommand{\olaxiom}{\mathsf{OL}}

\usetikzlibrary{arrows,decorations,shapes,automata,positioning,decorations.pathmorphing}
\tikzset{
  every picture/.style = {
    thick,
    >=stealth',
    node distance = 1.3em and 2.5em,
  }
  ,
  cross line/.style = {
    preaction = {
      draw=white,
      -,
      line width=4pt
    }
  }
  ,
  state/.style = {
    rectangle,
    rounded corners = 5pt,
    font = \footnotesize,
    draw,
    minimum width = 1em,
    minimum height = 1em
  }
  , 
  label-state/.style = {
    sloped,
    font = \scriptsize,
    label distance = -2.5pt
  }
  , 
  label-edge/.style = {
    font = \scriptsize,    
    label distance = -2.5pt
  }
}

\makeatletter
\def\thm@space@setup{%
  \thm@preskip=0.10cm
  \thm@postskip=\thm@preskip
}
\makeatother

\title{Beyond the Spell: A Dynamic Logic Analysis of Misdirection}
\titlethanks{}
\titlerunning{Beyond the Spell: A Dynamic Logic Analysis of Misdirection}

\addauthor{Benjamin Icard}{LIP6, Sorbonne University, CNRS, France}
\addauthor{Raul Fervari}{Consejo Nacional de Investigaciones Cient\'ificas y T\'ecnicas (CONICET) and Universidad Nacional de Córdoba, FAMAF, Argentina}
\authorrunning{Icard and Fervari}

\begin{document}
\maketitle

\begin{abstract}
Misdirection can be defined as the intentional action of causing some misrepresentation in an agent, or in a group of agents. Such misrepresentations may result from verbal actions, as in linguistic deception, or from visual actions, as in visual misdirection. Examples of visual misdirection abound (e.g. in nature, in the military), with magic tricks providing a vivid illustration. So far, various types of verbal misdirection have been investigated from a formal perspective (e.g. lying, bluffing) but little attention has been paid to the particular case of visual misdirection. In this paper, we introduce a dynamic epistemic logic to represent not only verbal misdirection on agents' beliefs but also visual misdirection on agents' observations. We illustrate the dynamics of the logic by modelling a classic magic trick known as the French Drop. We also provide a sound and complete axiom system for the logic, and discuss the expressivity and scope of the setting.
\end{abstract}

\keywords{Observations, Beliefs, Misdirection, Dynamic Epistemic Logic, Verbal, Visual, Simulation, Dissimulation, Surprise}




\section{Introduction}
\label{sec:introduction}

Human misdirection refers to an agent’s intentional action to make addressees form inaccurate representations of the world by “\textit{deflecting [their] attention for the purpose of disguise}” \cite{sharpe1988conjurers}. These inaccurate representations, or \textit{misrepresentations}, can arise through different modalities, either strictly verbal (e.g., announcements) or non-verbal (e.g., visual cues) \cite{Bell&Whaley1991,kuhn2014psychologically}. When the modality is verbal, as in announcements, misrepresentations usually consist in inaccurate beliefs (with the notion of belief defined as in~\cite{Wright51}), that can be associated with observations, as in so-called ``\textit{epistemic seeing}''~\cite{Dretske1970,Dretske1979}. Therein, it is also established that when the modality is non-verbal, as in visual observation, misrepresentations consist in inaccurate observations that can involve inaccurate beliefs (i.e., inaccurate \textit{epistemic seeing}) but not necessarily, as in so-called ``\textit{non-epistemic seeing}''. To the best of our knowledge, no formal treatment has been given to visual misdirection and resulting inaccurate observations, in contrast to verbal misdirection. To fill this gap, we propose to use dynamic epistemic logic (DEL) to analyze both verbal and visual misdirection, their similarities and differences, and resulting misrepresentations. For clarity, we call “strategies of misdirection'' the pairings of actions with the misrepresentations they are designed to produce.  

In the conceptual and logical literature, misdirection has been mainly studied through linguistic lenses, i.e. as verbal deception. As the prototypical form of verbal deception~\cite{Carson2010,Mahon2015}, lying has received the most analytical attention (e.g., \cite{Sakama&al2010,Vanditmarsch&al2012}). However, related borderline strategies, such as \textit{bluffing}~\cite{Sakama2015,Vanditmarsch2014}, \textit{half-truths}~\cite{Sakama&Caminada2010}, and \textit{omission of information}~\cite{Sakama&Caminada2010,Sakama&al2014}, have also been formally analyzed. In the course of those investigations, visual strategies of misdirection have been largely overlooked, – even though the concept of visual observation has caught some attention~\cite{Charrier&al2016,Herzig&al2015}.


Visual misdirection is pervasive in humans \cite{Bell&Whaley1991,forbes2011dazzled} as well as in nature \cite{quicke2017,ruxton2019avoiding}.
A fundamental distinction concerning misdirection is between \textit{dissimulation} and \textit{simulation}. In~\cite{Bell&Whaley1991}, dissimulation is defined as the action of hiding facts to deceive addressees, whereas simulation consists in deceiving them by showing some unreal facts as if they were really happening. Important to notice is that, in nature, camouflage and faking embody this broader distinction between \textit{dissimulation} and \textit{simulation} in misdirection.  
In \emph{camouflage}, one individual aims to blend into its environment to remain undetected from predators, while \emph{faking} consists of simulating a false or counterfactual reality to distract their attention. 
Dissimulation and simulation actually occur in all living entities. In humans, those strategies have been practiced since ancient times by the military, e.g., \textit{simulation} with the Trojan horse \cite{whaley1969stratagem}. But, actually, stage magicians have also played a central role, even in the military domain \cite{maskelyne1949magic} by crafting visual illusions to trick enemies, based on \textit{simulation} (with fake tanks) or on \textit{dissimulation} (with light effects)~\cite{fisher2011war}.


\paragraph{\bf Main contributions.} This paper aims to deepen the understanding of misdirection through a logical analysis of its verbal and visual forms. Rather than proposing a new theory of misdirection, we formalize established conceptual distinctions, including the contrast between non-epistemic and epistemic seeing~\cite{Dretske1970,Dretske1979} and that between visual simulation and dissimulation~\cite{Bell&Whaley1991}. Our goal is to clarify, from a logical perspective, how these distinctions interact and what role each plays in the production of misdirection. To this end, we analyze the French Drop, a classical magic trick (see, e.g.~\citep{bell2003toward,kuhn2014psychologically}), in which a performer appears to transfer a coin from one hand to the other while both simulating the transfer and dissimulating the coin. Logical work on magic has already shown how formal tools can illuminate magical effects, for instance by analyzing the sense of impossibility they produce~\cite{smith2016construction}. We likewise address the cognitive effect of surprise in the French Drop, but our main objective is to characterize the deceptive structure of misdirection, verbal and visual, in magic and beyond. We consider this a starting point in modeling different forms of misdirection and their properties.

To carry out this analysis, we introduce a dynamic logic of misdirection, $\dol$, for truthful and deceptive verbal and visual actions. Technically, $\dol$ adapts standard dynamic epistemic logic~\cite{BaltagMS98,DELbook} by combining atoms for visual observation, a doxastic modality for belief, and a dynamic modality for updating beliefs and observations under verbal announcements and visual actions. Our treatment of visual observation is inspired by the logic of (in)attentive agents~\cite{BB23} (see also~\cite{BolanderDHLPS16}), but departs from it in two respects: observing $p$ neither amounts to attending to $p$ nor entails believing $p$, and we distinguish syntactically between observing $p$ and observing $\neg p$. As in~\cite{BB23}, not observing a fact does not imply observing its negation, so the relevant models must satisfy suitable constraints. The dynamic component relies on action models with postconditions~\cite{DK08,DitmarschHL12}, including models in which one designated agent misdirects the others. This framework allows us to track, step by step, how actions affect belief and visual observation. It also lets us formalize not only the French Drop, with emphasis on visual simulation, but also related notions such as dissimulation, stronger notions of belief and observation, and surprise as a cognitive effect. 
Moreover, even though the original logical machinery we introduce is rather simple, we argue that this work still constitutes a novel, interesting application of the Dynamic Epistemic Logic family.

\paragraph{\bf Outline of the paper.} In Section~\ref{sec:concepts}, we discuss the concepts which motivate our proposal and technical approach. We review the logical literature on the formalization of misdirection, and put emphasis on missing features in extant proposals. Section~\ref{sec:definitions} introduces our dynamic logic of misdirection called $\dol$.
This section also provides a sound and complete axiom system for $\dol$ over the appropriate class of models. In Section~\ref{sec:types}, we show that a small fragment of the logic called $\tslogic$ is sufficient for our expressive purposes, since dynamic modalities only make use of four (types of) action models to capture genuine and deceptive actions, either verbal or visual. Section~\ref{subsec:example} illustrates the relevance of $\dol$ by modelling the dynamics of the French Drop trick. In Section~\ref{sec:other}, we show that our setting can express stronger belief attitudes, epistemic versions of observation, more complex notions of misdirection and also the sense of surprise resulting from magic misdirection.  
Section~\ref{sec:final} provides some final remarks and future directions of research.


\section{Setting the Stage for Misdirection}
\label{sec:concepts}

We discuss here key distinctions concerning agents' beliefs and observations. Then, we review extant logical frameworks for formalizing misdirection against the background of the classical distinction between \emph{dissimulation} and \emph{simulation}. Finally, we identify the logical features required to capture verbal and visual misdirection.


\subsection{Believing and Observing}

For simplicity, we distinguish two perspectives, that of the operator, who acts on the basis of some information, and that of the addressee, who reacts to the information received. Among possible modalities, we focus on verbal and visual actions, setting aside others such as auditory, kinesthetic, or olfactory ones. A verbal case is that of a politician addressing potential voters so as to make them believe that he or she is the best candidate. A visual case is that of a stage magician performing for spectators so as to make them observe an unreal event, such as the appearance or disappearance of a rabbit. In verbal contexts, the operator and addressee are usually called the ``speaker'' and the ``hearer''. In visual contexts, the operator may be called the ``performer'', and the addressees the ``spectators'', ``audience'', or ``public''.

%
%

Let us start by drawing distinctions concerning verbal and visual attitudes before considering the actions affecting those attitudes. Agents can entertain three possible, yet non exclusive, attitudes with respect to their surroundings:

\begin{itemize}
\item \textbf{\textit{Observing with believing}}, called here ``epistemic observation'', corresponds to Dretske’s notion of ``epistemic seeing''~\cite{Dretske1970,Dretske1979}. Some objects, such as an apple or a parrot, are not merely seen as shapes but seen \textit{as} an apple or \textit{as} a parrot. Likewise, a detective may observe pieces of evidence not as scattered objects and marks, but as a crime scene indicating a struggle or a crime. The same applies to facts or events. An agent epistemically sees a car moving forward in the street when she sees it \textit{as} a car moving forward, rather than as an unidentified moving object. In such cases, observation involves conscious recognition and yields explicit beliefs about what is observed.

\item \textbf{\textit{Believing without observing}}, called herein ``simple belief'', concerns cases in which agents believe that certain events occur beyond their sight without observing them directly. For example, an agent may believe that a car has crashed on the street on the basis of verbal information from a direct observer, even though she has not witnessed the event herself.



\textbf{\textit{Observing without believing}}, called herein ``atomic observation'', corresponds to Dretske’s notion of ``simple seeing''~\cite{Dretske1970,Dretske1979} and to ``non-epistemic seeing'' in~\cite{Smith2001,Demircioglu2017}. It can be understood as a causal process. When an object, such as an apple or a parrot, is in an agent’s field of vision and light reaches the eyes, it may be perceived without being recognized \textit{as} an apple or \textit{as} a parrot. Likewise, a person looking through a window may observe the street in this atomic, that is, non-epistemic, way if she forms no explicit doxastic representation of what occurs there, for instance of a parrot flying or a car moving forward, because she remains absorbed in her thoughts.

\end{itemize}

To sum up, belief necessarily involves a doxastic stance toward the world, whereas observation does not~\cite{Dretske1970,Dretske1979}. In \textit{epistemic} observation, the agent does adopt such a stance, since she consciously recognizes what she observes and thereby forms beliefs about it. In \textit{atomic}, i.e., \textit{non-epistemic}, observation, by contrast, the agent forms no belief about what she sees, as when she does not attend to what is happening in her immediate surroundings.


Regarding actions, verbal modalities involve acts of \textit{telling} facts to addressees, whereas visual modalities involve acts of \textit{showing} them. In both cases, successful actions produce epistemic effects in addressees, in different ways. Showing leads addressees to observe facts or events \textit{epistemically}, not merely to see them. The act of showing directs their attention to what is presented, and thereby prompts explicit beliefs about it.Telling, by contrast, leads either to simple belief, when addressees cannot observe the relevant facts themselves, or to epistemic observation, when they can. In both cases, the addressee acquires new beliefs that change her uncertainty about the world. These distinctions are summarized in Figure~\ref{tab:perspectives}.

\begin{figure}[t]
\begin{center}
\begin{tikzpicture}[scale=.86]
    \draw[thick] (0,0) circle (0.7cm);
    \node at (0,0) {\scriptsize Telling};
    \draw[thick] (2.1,0) circle (0.7cm);
    \node at (2.1,0) {\scriptsize Showing};
    \draw[thick] (0.3,-3.4) circle (1.85cm);
    \node at (-0.6,-3.5) {\scriptsize\shortstack{Simple \\belief}};
    \draw[thick] (2.1,-3.5) circle (1.85cm);
    \node at (3.05,-3.4) {\scriptsize\shortstack{Atomic \\ observation}};

    \node at (1.2,-3.4) {\scriptsize\shortstack{Epistemic \\ observation}};

    \draw[decorate,decoration={brace,amplitude=8pt},thick] (-0.5,1.2) -- (2.6,1.2);
    \node at (1,1.7) {\footnotesize\textbf{Types of Modalities}};
    \node at (0,1) {\footnotesize\textit{Verbal}};
    \node at (2.1,1) {\footnotesize\textit{Visual}};

    \draw[decorate,decoration={brace,amplitude=8pt,mirror},thick] (-2.2,1.3) -- (-2.2,-4.5);
    \node[rotate=90] at (-2.9,-1.4) {\scriptsize\textbf{Perspectives of Agents}};
    \node[rotate=90] at (-1.9,0) {\scriptsize\textit{Operator's action}};
    \node[rotate=90] at (-1.9,-3) {\scriptsize\textit{Addressee's reaction}};

    \draw[->, thick, dashed] (0,-0.7) -- (0.3,-1.6);
    \draw[->, thick, dashed] (2.1,-0.7) -- (1.05,-2.5);

\end{tikzpicture}
\end{center}
\caption{Possible interactions between agent's perspectives, or attitudes, regarding verbal versus visual informational modalities.}
\label{tab:perspectives}
\end{figure}

Operators also separate into two categories regarding their intentions to addressees: \textit{genuine} versus \textit{misleading}. Intentions are genuine when the operators' goal is to inform addressees, while their intentions are misleading when they are deceptive and plan to mislead addressees. This is the case of misdirection. From the perspective of actions, misdirection can be verbal or visual. Verbal misdirection consists in telling addressees something false, whereas visual misdirection consists in presenting them with a false representation of reality.

\subsection{Two Main Strategies of Misdirection}
\label{subsec:twomain}

Building on~\cite{Bell&Whaley1991}, we formalize two central strategies of misdirection: \textit{dissimulation of the truth} and \textit{simulation of the false}. We treat them in their most general form, distinguishing only between their visual and verbal manifestations. For completeness, this section also reviews extant work on their main subcategories, to indicate the range of cases covered by our formal approach.

\subsubsection{Misdirection as Dissimulation of the Truth} 
\label{subsubsec:dissimulation}


Dissimulation consists in hiding the truth to addressees to prevent them from making an accurate representation of the world. For e.g.~\cite{Bok1983,Fallis2020}, hiding, or concealing information, consists in \textit{masking} true information, — aiming addressees not to form accurate beliefs on the world. Important to notice is that dissimulation in this sense (i.e., hiding or concealing), differs from dissimulation as \textit{withholding information}, or omission. The work in~\cite{Carson2010} insists on the fact that withholding information does not involve any effective action to deceive, contrary to hiding information which requires an effective action of concealment in that respect. 

The effect of successful dissimulation is always \textit{epistemic}: the addressee fails to reach accurate belief and/or accurate epistemic observation on her surroundings. In the verbal case, successful dissimulation leads addressees \textit{not} to reach true beliefs, and can also lead them \textit{not} to reach accurate epistemic observation when speakers use verbal omission to avoid attracting the visual attention of addressees to some specific fact or event. In the visual case now, successful dissimulation necessarily leads addressees to fail to epistemically observe some true reality. This is the case in the French Drop trick, where dissimulation is visual and relies on \textit{masking}. As part of the trick, the magician conceals the coin by \textit{palming} it into her left hand. In doing so, she puts a physical obstacle, namely her left hand, between the eyes of the public and the coin, thus concealing it from their sight. This makes spectators fail to epistemically observe that the coin is (still) in the magician's left hand. We argue that in the French Drop, this dissimulation through palming occurs simultaneously with the action of simulating a fake pass, which enhances the magical effect by diverting the audience's attention.

%

Verbal dissimulation has been formalized from two different perspectives. The first one has focused on expressing the \emph{epistemic preconditions} of dissimulation~\cite{Sakama2015,Vanditmarsch&al2012,Vanditmarsch2014}. By default, the precondition for dissimulation is the existence of a mismatch between the level of uncertainty the speaker has about the world, and the level of uncertainty he or she wants to convey to addressees. A variant on dissimulation is bluffing \cite{Sakama2015,Vanditmarsch2014}: now, speakers are completely uncertain about the truth or falsity of the utterances they make. But, as a precondition, the strongest form of uncertainty consists in full ignorance: in this case, speakers do not know, or at least believe to be true, the content they are uttering. Epistemic logicians have modelled different forms of ignorance but not from the sole perspective of deception~\cite{Vanderhoek2004,Fan&al2015,kubyshkina2021logic}.    

The second perspective for modeling verbal dissimulation is more concerned with the \emph{epistemic goals} of the deceivers. Following \cite{Chisholm&Feehan1977}, \cite{Sakama&Caminada2010} characterizes deception by {\em commission} in which deceivers perform some action to hide the truth, is distinguished from deception by {\em omission} in which deceivers do not perform any action to this end. Omission can be either \emph{complete} or only \emph{partial}. When it is complete, deceivers intend to hide the whole truth to addressees and thus, to prevent the formation of \textit{any} belief in them. In this regard, \cite{Sakama&Caminada2010,Sakama&al2014} have formalized the epistemic goals of complete omission under the name of \textit{``withholding information''}. This strategy should be contrasted with the utterance of \textit{``half-truths''} in which dissimulation is only partial. Following \cite{Egre&Icard2018,Egre2021}, half-truths are vague utterances whose status is deliberately uncertain between truth and falsity, and that speakers may use to make addressees draw mistaken interpretations and conclusions. The epistemic goal of such half-truths is to keep addressees in the dark, as it is modeled in e.g. \cite{Sakama&al2014}.

By contrast with verbal dissimulation, visual dissimulation has not been logically investigated so far, though the notions of observation and perception have caught logical attention in e.g. \cite{Herzig&al2015,Charrier&al2016,BB23}. In visual dissimulation, one attempts to escape the gaze of others based on using various means, detailed in the taxonomy of deception proposed by \cite{Bell&Whaley1991}. One of those means, called \textit{``masking''}, consists in preventing any potential observation by putting an obstacle between oneself (or the object to hide) and the observer in order to remain invisible (or the object). Another mean used for visual dissimulation is  \textit{``camouflage''}, or \textit{``crypsis''}, in which the deceiver blends into its surroundings by adequately matching its colours and/or shapes. To the best of our knowledge, none of these distinctions has been formalized so far. 

\subsubsection{Misdirection as Simulation of the False}  
\label{subsubsec:simulation}

In contrast with dissimulation, simulation consists in presenting false information to addressees, either by using available false information or by fabricating false information, as in \textit{faking} or \textit{fabrication} \citep[]{Fallis2009b}. The goal of simulation is to entertain addressees with inaccurate representations of the world. 

As for dissimulation, the effect of successful simulation is always \textit{epistemic} too, consisting in wrong beliefs and/or wrong epistemic observations for addressees. In the verbal case, successful simulation leads addressees to false beliefs, but can also lead them to false epistemic observation when addressees also see the physical manifestation \textcolor{black}{of the simulated action.} In the visual case now, successful simulation converts atomic observation into false epistemic observations for the addressees. This is the case in the French Drop trick: before the magician performs any visual action, the spectators observe the stage atomically since they do not necessarily pay attention to any specific object or fact of the scene. They just expect the magic trick to begin. Then, the magician enters the scene and the trick starts: she dissimulates the coin by palming into her left hand and then simulates the false action of moving the coin from left to right hand. By the magician performing those visual actions, the spectators' atomic observations become epistemic observations. At the end of the effect, indeed, the audience believe and epistemically observe that the coin has moved to the magician's right hand when this is not the case.



Extant formal literature on verbal simulation has given particular attention to the act of \textit{lying}, as a prototypical instance of verbal simulation in which speakers utter propositional contents they believe to be false with the intention to deceive addressees. Intention (to deceive) has been formalized in static epistemic logic with classical Kripke models \cite{Sakama&al2010,Sakama&al2014}, while their assertoric and deceptive dimensions have been expressed in dynamic epistemic logic \cite{Liu2009,Vanditmarsch&al2012,Vanditmarsch2014,Agotnes&al2016}. Using announcement operators to express intentionally false utterances, as in~\cite{Liu2009}, \cite{Vanditmarsch&al2012,Vanditmarsch2014} logically distinguish \textit{“subjective lies”} (i.e. believed to be false utterances that happen to be true in reality) from more \textit{“objective lies”} (i.e. believed to be false utterances that also turn out to be really false). In \cite{Agotnes&al2016}, more exotic cases of lying are formalized, such as Moorean lies, i.e. utterances that are initially false but become true once they are announced. Apart from lies, verbal simulation can also consist in making addressees perform \emph{wrong inferences} based on the speakers' utterances. In particular, \cite{Sakama&al2010} formalizes \textit{“deductive lies''} in which speakers want addressees to infer inaccurate conclusions from their utterances, distinguishing them from \textit{“abductive lies''} in which speakers expect them to jump directly to false conclusions, without any deduction \citep{Sakama&Katsumi2016}. 

Here again, by contrast with verbal simulation, visual simulation has not been logically investigated to the best of our knowledge. As mentioned, only the notions of observation and perception underlying visual simulation have caught logical attention recently \cite{Herzig&al2015,Charrier&al2016,BB23}. 
In turn, our interest is to model these two concepts in general via dynamic actions changing the perception of an agent, incorporating also the notion of observation into the picture as a novelty.

\subsection{Modelling Verbal and Visual Misdirection} 

\subsubsection{Misdirection in the French Drop Trick}

Misdirection has been analyzed extensively by stage magic practitioners (see e.g.~\citep{maskelyne1911our,fitzkee1975magic}). When magicians aim to master a new trick, they often break it down into two key components: the \textit{method} used to perform the trick, and the \textit{effect} that results from the method (for a detailed discussion on the stages involved in producing magical effects, see \cite{sharpe1932neo}). Crucially, the method refers to the {\em action}, or {\em combination of actions}, employed by magicians to lead spectators into a deceptive state, known as the {\em effect} (see e.g.~\cite{Kuhn&al2016}).

 In magic tricks, spectators are deceived because they are disoriented by the magician's method. They expect the magician's action to result in some state of affairs but, as a result of the magician's action, they observe that another state of affairs obtain which contradicts their expectations. At the heart of the magician's deceitful effect is a mismatch between the spectators' expectations and the result they actually observe. Spectators expect the magician to do some particular action but in fact, the magician has not played this action but another one resulting in deception and surprise.  
%
%
Here, we use the French Drop Trick (see \citep{bell2003toward,kuhn2014psychologically}) to formalize and illustrate visual actions of misdirection. Accordingly, states of deception in the French Drop concern the method and, thus, the effect.


\subsubsection{Casting the magic spell \textit{logically}}

We aim to define actions of misdirection of two distinct types. The first type is verbal and consists, for an operator, to induce a false belief in an addressee by making a linguistic utterance. The second type is visual and now consists to trigger a wrong observation \textcolor{black}{(i.e., observing something that differs from reality)} in the addressee by making gestures. Those types act at two different levels in addressees. In verbal misdirection, the operator's announcement acts at the level of the addressees' beliefs to make them believe the false, or at least to prevent them from believing the truth. In visual misdirection, the operator's gestures act at the level of the addressees' observations and beliefs, to make them see something that is not the case, as it happens in visual illusions for instance.      

In this paper, the logic $\dol$ will contain a static operator for expressing agents' beliefs as well as special atomic formulas for expressing their (non-epistemic) observations. But $\dol$ will also contain a dynamic operator to capture the agents' informational actions. Semantically, we define specific action models for verbal actions, or actions of telling, and specific action models for visual actions, or actions of showing. In both cases, actions are either genuine or misleading.
       
Our setting turns out to be powerful enough to distinguish between simulation and dissimulation, in the verbal and visual cases. We argue that it can also express other notions such as epistemic observation and stronger notions of belief and observation. Our logic shows that visual actions convert atomic observation in epistemic observation but also that that simulation (of the false) implies dissimulation (of the truth). Finally, we use $\dol$ as a tool to put more emphasis on the specificity of misdirection in magic tricks compared to other forms: the state of surprise caused by the deception involved. That way, we aim to capture the essence of misdirection in magic tricks as a novel application of Dynamic Epistemic Logic, but also to pave the way for using $\dol$ to analyze other forms of misdirection in the future.


\section{A Dynamic Logic of Misdirection}
\label{sec:definitions}

We present the syntax and semantics of our \emph{Dynamic Logic of Misdirection}~($\dol$). We introduce a class of models called herein \emph{observational epistemic models} that contain two kinds of information: about beliefs, which can result from both verbal communication and visual perception, and information about the facts that are observable for the agents, which result from visual perception only. While beliefs will be represented in the language via a standard doxastic modality~\cite{Hintikka:kab}, observations will be modelled by special atomic symbols, similarly to the way~\cite{BB23} models attention.  
Unlike other approaches like \cite{Herzig&al2015,Charrier&al2016}, our framework will allow us to model dynamic aspects of beliefs and visual misdirection. We also discuss the appropriateness of certain classes of models and present a standard axiomatization for the logic.

\subsection{Syntax and Semantics}
\label{subsec:synsem}

The language presented herein is an adaptation of the Dynamic Epistemic Logic (DEL) for (in)attentive agents presented in~\cite{BB23}, based on generalizing~\cite{BaltagMS98}. 
Throughout the text, we let $\prop$ be a countable set of propositional symbols, $\agt$ a finite set of agent symbols, and $\obsym=\{\obsp,\obsnp \mid a\in\agt, \ p\in\prop\}$ a set of observation symbols. We read $\obsp$ (respectively $\obsnp$) as \emph{``agent $a$ observes that $p$ (respectively~$\neg p$) is the case''}. Notice that in~\cite{BB23}, some similar extra symbols are used to state that an agent is paying attention to a certain atomic proposition $p$. But here, the conceptual difference is that an extra atomic symbol $\obsp$ is added to stand for non-epistemic observation, i.e., for the fact that an agent $a$ only visually perceives that some fact $p$ is happening, possibly without paying attention to it. We start out by formally introducing the models we use to interpret our logic.

\begin{definition}\label{def:models} A \emph{relational (Kripke) model} is a tuple $\model=\tup{W,\{R_a\}_{a\in\agt},V}$ where: 
\begin{itemize}
    \item $W$ is a non-empty set of \emph{possible worlds};
    \item $R_a\subseteq W\times W$; 
    \item $V: \prop\cup\obsym \ra 2^W$ is a \emph{valuation function}.
\end{itemize}

Let $w\in W$, we call the pair $(\model,w)$ a \emph{pointed model}, with parentheses usually dropped. 
\end{definition}

Below we introduce a particular class of models that turns out to be interesting in the rest of the paper. Intuitively, in this class, we force the relations $R_a$ to have the classical properties of \emph{belief}~\cite{Wright51,Hintikka:kab}. For observations, 
we need to impose conditions to guarantee that they are \emph{consistent}. 
We will come back to this point later on the paper, after presenting the whole framework. 

\begin{definition}\label{def:obs-epis-model}
We define the class of \emph{observational epistemic models} as the set of all the pointed models $\model,w$ such that $\model=\tup{W,\{R_a\}_{a\in\agt},V}$, $w\in W$, and where: 

\begin{itemize}
    \item each binary relation $R_a$ is Euclidean, transitive and serial; and
    \item for all $p\in\prop$ and $a\in\agt$, $w\in V(\obsp)$ implies $w\notin V(\obsnp)$.
\end{itemize}
\end{definition}

At this point, an important distinction to make concerns the models used for (in)attentive agents in~\cite{BB23}, compared to the ones we use here. In \cite{BB23}, for each propositional symbol $p$, and each agent $a$, there is a special propositional symbol $\mathsf{h}_ap$ that labels a world in case the agent $a$ is attentive to symbol $p$ at such a world. In our setting, extra symbols for observation, written $\obsp$ (with $p\in\prop$), are no longer indicating ``attentiveness'', but they indicate that agent $a$ is \emph{observing} that a certain fact $p$ holds, without necessarily paying attention to it. As already mentioned, the agent observes fact $p$ in a non-epistemic way in this case (see e.g.~\cite{Dretske1970,Dretske1979,Smith2001,Demircioglu2017}), since the agent does not endorse any doxastic stance regarding the fact observed (e.g. when we see an entity in the water without forming any belief about it, for instance that it is a fish). In turn, we introduce two symbols instead of one: $\obsp$ to indicate that the agent observes that $p$ is the case, and $\obsnp$ to indicate that the agent observes that $p$ is not the case. As one can notice, $\obsp$ being false is not the same as $\obsnp$ being true since they can both be false but they cannot both be true at the same time. Finally, notice that each relation $R_a$ is Euclidean, transitive and serial, i.e., has the usual (basic) properties of belief~\cite{Wright51,Hintikka:kab}.
 
The actions of misdirection that some agent(s) perform will be represented via standard tools of Dynamic Epistemic Logic known as \emph{action models}~\cite{BaltagMS98}.
Recall that action models are objects describing semantically the effects intended by an agent over an epistemic situation (i.e., over a model). They can be used as parameters in formulas of the language, in particular, inside a dynamic modality. In a nutshell, dynamic modalities perform updates in the current model, in this case depending on the semantic description given by an action model. When interpreting the dynamic modality, action models tell us how to update the original model, raising a new epistemic situation. 
In Definitions~\ref{def:actionmodels} and~\ref{def:syntax}, action models and the language are defined as usual, by mutual recursion.

\begin{definition}\label{def:actionmodels}
An \emph{action model} is a tuple $\amodel=\tup{E,\{\ra_a\}_{a\in\agt},\pre,\post}$, where:
\begin{itemize}
    \item $E$ is non-empty finite set of epistemic \emph{actions} that, conceptually, can be either of verbal or visual type;
    \item $\ra_a\subseteq E\times E$, for each $a\in\agt$, is Euclidean, transitive and serial;
    \item $\pre: E \ra \dol$ is the \emph{precondition function}, where $\dol$ is the language given by Definition~\ref{def:syntax}; and 
    \item $\post: E \ra (\prop\cup\obsym) \ra \dol$,  the \emph{postcondition function}, is such that for a finite subset $D\subseteq\prop\cup\obsym$, we have:
    \begin{itemize}
        \item $\post(e)(p)\in\set{\top,\bot}$, for $p\in D$, 
        \item for all $\set{\obsp,\obsnp}\subset\obsym$, we have $\obsp\in D$ if and only if $\obsnp\in D$, and for all $e\in E$,
        \begin{itemize}
            \item $\post(e)(\obsp)=\top$ implies that $\post(e)(\obsnp)=\bot$, 
            \item $\post(e)(\obsnp)=\top$ implies that $\post(e)(\obsp)=\bot$,
        \end{itemize}
        \item for all $p\notin D$, and all $e\in E$, $\post(e)(p)=p$.
    \end{itemize}

\end{itemize}

Let $e\in E$, we call $(\amodel,e)$ a \emph{pointed action model} (parentheses usually dropped). 
\end{definition} 

Notice that in Definition \ref{def:actionmodels}, action models incorporate \emph{postconditions}, as in~\cite{DK08}, and unlike~\cite{BB23}. This variant does not add expressive power but simplifies the definition of actions of verbal and visual misdirection since the truth value of atomic observations can be updated. Postconditions are also required to update the models consistently, and respecting the properties of an observational model (see Lemma~\ref{lemma:preservation-updates}). In this regard, we only allow updates of the truth value of a finite set of symbols (either propositional or atomic observations) with $\top$ or $\bot$, and for the rest, $\post$ acts as the identity function. Moreover, we ask that if an observation $\obsp$ is updated, so is its counterpart $\obsnp$, and the update is defined accordingly. In the next definition, we make precise the syntax of our dynamic logic of misdirection.

\begin{definition}\label{def:syntax}
The set of well-formed formulas of the language $\dol$ is given by the following Backus-Naur Form:
\[
\varphi, \psi ::= p \mid \obsp \mid \obsnp \mid \neg\varphi \mid \varphi\wedge\psi \mid \B_a\varphi \mid [\amodel,e]\varphi,
\]
\noindent where $p\in\prop$; $a\in\agt$;
and $\amodel,e$ is a pointed action model, whose preconditions and postconditions are defined in a lower level of the recursive definition. Other Boolean connectives are defined as usual: $\bot:= p\wedge\neg p$, $\top:=\neg \bot$, $\varphi\vee\psi:=\neg(\neg\varphi\wedge\neg\psi)$, $\varphi\ra\psi:=\neg\varphi\vee\psi$ and $\varphi\lra\psi:=(\varphi\ra\psi)\wedge(\psi\ra\varphi)$. We also define the following abbreviations: $\Bdiam_a\varphi := \neg\B_a\neg\varphi$; and $\tup{\amodel,e}\varphi := \neg[\amodel,e]\neg\varphi$. 
The set of literals is defined as $\lit=\{\ell \mid \ell \in \prop \mbox{ or } \ell=\neg p \mbox{ with } p\in\prop\}$. If $\ell=\neg p$, we use sometimes $\neg\ell$ to denote $p$; and use $\obs{a}{\neg\ell}$ to denote $\obs{a}{p}$. 
Finally, we define the fragment~$\ol$ to be the language generated from $\dol$ when we remove formulas of the form $[\amodel,e]\varphi$.
\end{definition}

The modalities introduced in the BNF above should be read as follows. Formula~$\B_a\varphi$ means `agent $a$ believes that $\varphi$ holds'. Its dual is $\Bdiam_a\varphi$ meaning `$\varphi$ is believable for agent $a$', or, in other words, that `$\varphi$ is consistent with agent $a$'s beliefs'. Here, by taking inspiration from~\cite{BB23}, observations are not epistemic but factual. Thus, they are represented by special atomic propositions: $\obsp$ indicates that \emph{agent $a$ observes that $p$ is the case}, while $\obsnp$ states that \emph{agent $a$ observes that $p$ is not the case}. 
Definition~\ref{def:semantics} details the semantics of all the formulas of $\dol$ but, intuitively, beliefs are conceived as weak doxastic attitudes, or so-called \emph{prima facie} beliefs. Thus, when an agent $a$ believes that $\varphi$, namely $\B_a\varphi$, agent $a$ may lack any justification for believing $\varphi$. In Section~\ref{sec:other}, this notion of \textit{weak belief} is used to define a notion of \textit{strong belief} according to which agents hold some justification to support their beliefs. 
Finally, dynamic formulas $[\amodel,e]\varphi$ are read as `after the  action $\amodel,e$ is executed, $\varphi$ holds', while its dual $\tup{\amodel,e}\varphi$ is read as `the action $\amodel,e$ is executable at the current situation and after its execution $\varphi$ will hold'. 


Now, we define how an action model affects an observational epistemic model, — potentially modifying the doxastic states and observations of a set of agents. In general, these actions are acts of communication between agents which can consist in verbal announcements or visual gesture, and can be either genuine or deceptive. As in usual settings, such communication can be public, private or semi-private, but might involve both verbal and visual information herein.
Moreover, notice that while in~\cite{BB23}, actions of communication are \emph{truthful}, we consider both truthful and untruthful actions here, as we want to address different forms of communication. In addition, our action models are enriched with \emph{postconditions}, which stand as descriptions of how certain facts change after an action, following e.g.~\cite{DK08}. In our case, postconditions are used to express changes in the facts that an agent is observing, as well as changes in the agent's observations themselves.

\begin{definition}\label{def:product}
 Let $\model=\tup{W,\set{R_a}_{a\in\agt},V}$ be a relational model, and 
 let $\amodel=\tup{E,\{\ra_a\}_{a\in\agt},\pre,\post}$ be an action model. We define their product $\model\otimes\amodel=\tup{W',\{R'_a\}_{a\in\agt},V'}$, where:
 \begin{itemize}
     \item $W'=\set{(w,e)\in W\times E \mid \model,w\models\pre(e)}$;
     \item $(w,e) R'_a (v,f)$ if and only if $wR_av$ and $e \ra_a f$, for all $a\in\agt$;
     \item $V'(p) = \{(w,e) \mid \model,w\models \post(e)(p)\}$.
 \end{itemize}
 \end{definition}
 
With all these definitions at hand, we can now define the semantics of $\dol$.

\begin{definition}\label{def:semantics}
Let $\model=\tup{W,\{R_a\}_{a\in\agt},V}$ be a model, and let $w\in W$. The satisfaction relation $\models$ between $\model,w$ and formulas $\varphi$ is defined as follows:
\[ 
    \begin{array}{l@{\quad }c@{\quad }l}
    \model,w\models p & \iffdef & w\in V(p) \ \quad (p\in\prop\cup\obsym)\\
    \model,w\models\neg\varphi & \iffdef & \model,w\not\models\varphi \\
    \model,w\models\varphi\wedge\psi & \iffdef & \model,w\models\varphi \mbox{ and } \model,w\models\psi \\
    \model,w \models \B_a\varphi & \iffdef & \mbox{for all $v\in W$ such that } w R_a v, ~ \model,v\models\varphi \\
\model,w \models [\amodel,e]\varphi & \iffdef & \model,w\models\pre(e) \mbox{ implies }   (\model\otimes\amodel),(w,e)\models\varphi.
    \end{array}
\]
A formula $\varphi$ is satisfiable if there exists a pointed model $\model,w$ such that $\model,w\models\varphi$, and it is valid if $\model,w\models\varphi$, for all pointed models $\model,w$.
\end{definition}

\begin{remark}\label{rem:dualmodality}
The dual modality $\tup{\amodel,e}$ of $[\amodel,e]$ is interpreted as:
\[
\model,w\models\tup{\amodel,e}\varphi \ \iffdef \ \model,w\models\pre(e) \mbox{ and }   (\model\otimes\amodel),(w,e)\models\varphi.
\]
\end{remark}

Let us make some observations concerning the execution of an action in our setting. Here we are interested in particular \emph{``informational actions''}, in the sense that some agent $a$ performs an informative act which leads to update the beliefs and observations of the rest of the agents. 
%
%
If agent $a$ executes the action~$e$, this action does not change her own factual beliefs and visual observation of the environment. However, other agents' beliefs and visual observations might change, but not necessarily though. This means that the agent can execute an action with some informational purpose but she cannot guarantee that this action will modify the other agents' beliefs and/or visual observations. Consider the following example: a magician performs an action to misdirect the attention of a crowd of spectators. The magician knows that actual facts do not change during the trick, but what she does not know, and cannot guarantee, is whether the spectators will believe that actual facts \emph{did} change. In other words, the magician cannot be sure that her misdirection will succeed. 
In formula, for an informational action model $(\amodel,e)$,  we have that $\model,w\models\B_a\varphi$ does not imply $\model,w\models[\amodel,e]\B_a\varphi$. 

Regarding the modelling of magic tricks, the approach of \cite{smith2016construction} using Propositional Dynamic Logic (PDL) differs from the DEL approach of $\dol$ in terms of focus of interest. While both rely on modelling epistemic notions, such as beliefs and expectations, to express states of deceptions, the goal of $\dol$ is to better understand the nature of misdirection. In particular, we are interested in visual misdirection, so the logic encodes notions such as atomic observation, as a non-epistemic attitude, and visual action of observation, now as an epistemic attitude, in addition to classical verbal announcements. By contrast,  \cite{smith2016construction} puts more emphasis on the constructional aspects of conjuring magic.\footnote{Analyzing Martin Gardner's \textit{Turnabout} trick (see \cite{fulves1977big}) in that respect.} They aim to express the various forms of \textit{impossibility} spectators may experience  with magical effects, based on characterizing the evidence relationship between expected events and real occurring events, --such as similarity, perceptual equivalence, structural equivalence, and congruence. To some extent, however, our two perspectives meet since, as we will see, $\dol$ is used in subsection \ref{ssec:surprise} to express the effect of \textit{surprise} obtained with the French Drop trick.

\subsection{Axiomatic System}
\label{sec:axiom}

We start out by introducing the axiom system $\olaxiom$ for the static fragment $\ol$.

\begin{definition}\label{def:staticaxioms}
    The axiom system $\olaxiom$ contains all the axioms and rules in Table~\ref{tab:axioms:static}.
\begin{table}
\fbox{
\begin{small}
\begin{minipage}{0.5\textwidth}
\begin{tabular}{ll}
    $\axm{CPL}$ & All axioms from propositional logic \\
    $\axm{Obs}$ & $\vdash\obsp \ra \neg\obsnp$ \\
    $\axm{K_\B}$ & $\vdash\B_a(\varphi \ra \psi) \ra (\B_a\varphi \ra \B_a\psi)$ \\ 
    $\axm{4_\B}$ & $\vdash\B_a\varphi \ra \B_a\B_a\varphi$ \\
    $\axm{5_\B}$ & $\vdash\neg\B_a\varphi \ra \B_a\neg\B_a\varphi$ \\
    $\axm{D_\B}$ & $\vdash\B_a\varphi \ra \Bdiam_a\varphi$  \\
\end{tabular}
\end{minipage}
\begin{minipage}{0.42\textwidth}
\begin{align*}
        & & \vcenter{\infer[\axm{(MP)}]{\vdash\psi}{\vdash \varphi \quad {\vdash\varphi \to \psi}}}
        \\
        & \\
        & &\vcenter{\infer[\axm{(Nec_\B)}]{\vdash\B_a\varphi}{\vdash\varphi}} 
\end{align*}
\end{minipage}
\end{small} }
\caption{Axioms and rules for $\ol$.} \label{tab:axioms:static}
\end{table}
\end{definition}

Notice that in the presentation of the axiomatic system, we use axiom schemes instead of plain axioms, i.e., we use meta-variables for formulas instead of fixed propositional symbols 
(see e.g.~\cite{van2012logic} for details on the use of axiom schemes). Still, for simplicity's sake, we refer to axiom schemes just as \emph{axioms} in what follows.

As expected, axioms for $\B_a$ are classical for doxastic modalities. Axioms $4_\B$ and $5_\B$ make sense for beliefs obtained by verbal means or any kind of sensory activity other than observations: it makes sense that one believes what one actually believes but also believes that she disbelieves something when it is the case (i.e., they are the standard axioms of introspection). Axiom $\axm{D}_\B$ accounts for seriality, and means that what is believed by an agent must be consistent with the agent's beliefs. Concerning observation, we do not have axioms at the epistemic level, since our approach considers that agents may observe without necessarily believing. In our framework, an agent can observe that $p$, that is $\obsp$, without $p$ being true, as it happens to be the case in common illusions or hallucinations.
However, as expected, an agent cannot observe $p$ and $\neg p$ at the same time. What we need to establish is the inconsistency of observing a fact and its negation, given by axiom $\axm{Obs}$. Notice that this axiom is consistent with our basic requirement imposed on valuations in observational epistemic models (see Definition~\ref{def:obs-epis-model}). 

The proof of soundness and completeness follows standard arguments of modal logic (see, e.g.,~\cite{Chellas_1980,Blackburn&deRijke&Venema01,HML}). Soundness is straightforward, except for the axiom $\axm{Obs}$ which is novel. Notice that if a pointed observational epistemic model $\model,w$ is such that $\model,w\models\obsp$, it is necessarily the case that  $\model,w\not\models\obsnp$, because $w\in V(\obsp)$ implies $w\notin V(\obsnp)$, by the second condition in Definition~\ref{def:obs-epis-model}. Thus, the axiom $\obsp \ra \neg\obsnp$ is also valid. 

Following~\cite[Proposition 4.12]{Blackburn&deRijke&Venema01}, given an arbitrary consistent set of formulas $\Gamma$,  completeness is obtained by building a model (known as the `canonical model') satisfying all the formulas in $\Gamma$. Associated notions like derivability (via $\olaxiom$), theoremhood, maximal consistent sets (written MCS) and consistency/inconsistency are defined in the usual way.

\begin{definition}[Canonical Model]\label{def:canonical}
    The \emph{canonical model} $\model^c=\tup{W^c,\{R^c_a\}_{a\in\agt},V^c}$ is defined as:
    \begin{itemize}
        \item $W^c := \{ \Delta \mid \Delta \mbox{ is an MCS of } \olaxiom \}$,
        \item $R^c_a := \{ (\Delta_1,\Delta_2)\in W^c\times W^c \mid \mbox{ for all $\psi$ in $\ol$}, \B\psi\in\Delta_1 \mbox{ implies } \psi\in\Delta_2 \}$, for each $a\in\agt$, and 
        \item $V(p) := \{\Delta \in W^c \mid p\in \Delta\}$, for each $p\in\prop$.
    \end{itemize}
\end{definition}

The canonical model will be adequate as a model for any consistent set of formulas. But since this model is built over \emph{maximal} consistent sets, the next classical lemma needs to be established. 

\begin{lemma}[Lindenbaum Lemma]\label{lemma:lindenbaum}
    Let $\Gamma$ be a consistent set. Then, there exists an MCS $\Gamma'$ such that $\Gamma\subseteq\Gamma'$. 
\end{lemma}
\begin{proof}
    See e.g.~\cite[Lemma 4.17]{Blackburn&deRijke&Venema01}.
\end{proof}

Moreover, the canonical model defined above is adequate for our purpose, only if it is an observational epistemic model.

\begin{lemma}\label{lemma:canonical-observ}
    The canonical model $\model^c=\tup{W^c,\{R^c_a\}_{a\in\agt},V^c}$ from Definition~\ref{def:canonical} is an observational epistemic model.
\end{lemma}

\begin{proof}
    In order to prove this proposition, we have to show that each relation $R^c_a$ is transitive, Euclidean and serial, and that for all $p\in\prop$ and $a\in\agt$,  $\Delta\in V^c(\obsp)$ implies $\Delta\notin V^c(\obsnp)$. The properties of $R^c_a$ are guaranteed by axioms $\axm{4}_\B$, $\axm{5}_\B$ and $\axm{D}_\B$ (see~\cite[Section~4.3]{Blackburn&deRijke&Venema01}). Suppose that $\Delta\in V^c(\obsp)$, for some MCS $\Delta$, $a\in\agt$ and $p\in\prop$. By definition of $V^c$, $\obsp\in\Delta$, and by $\axm{Obs}$, $\obsp\ra\neg\obsnp\in\Delta$. Then, by $\axm{MP}$, $\neg\obsnp\in\Delta$. By classical properties of MCS, $\obsnp\notin\Delta$, which proves the property.
\end{proof}

Now we can link the notion of derivability in $\olaxiom$ and the one of satisfiability (in the canonical model). 

\begin{lemma}[Truth Lemma]\label{lemma:truth}
    Let  $\model^c=\tup{W^c,\{R^c_a\}_{a\in\agt},V^c}$ be the canonical model from Definition~\ref{def:canonical}, and let $\Delta\in W^c$. Then,
    for all formulas $\varphi$ of $\ol$:
    \[
        \varphi\in\Delta \mbox{ if and only if } \model^c,\Delta\models\varphi. 
    \]
\end{lemma}

The proof of Lemma \ref{lemma:truth} proceeds by induction on the structure of $\varphi$. The case where $\varphi$ is of the form $\B_a\psi$, is obtained using the inductive hypothesis and the definition of $R^c_a$.  Once this property is established, completeness can be proved.

\begin{theorem}\label{th:completeness}
The axiom system $\olaxiom$ is sound and strongly complete with respect to observational epistemic models.
\end{theorem}

\begin{proof}
Suppose $\Gamma$ is a consistent set of formulas. By Lemma~\ref{lemma:lindenbaum}, there exists a MCS $\Gamma'$ such that $\Gamma\subseteq\Gamma'$. By Lemma~\ref{lemma:truth}, $\model^c,\Gamma'\models\Gamma$. Moreover, by Lemma~\ref{lemma:canonical-observ}, $\model^c$ is an observational epistemic model, which completes the proof.
\end{proof}

Now it is time to consider the full language $\dol$, i.e., to handle dynamic modalities as well. In order to axiomatize $\dol$, it is only needed to show that every formula $[\amodel,e]\varphi$ of $\dol$ can be transformed into another formula $\varphi'$ of $\ol$, i.e., with no occurrences of the dynamic modality $[\amodel,e]$. This transformation is possible since dynamic modalities do not add expressive power to $\ol$ (see e.g.~\cite{DELbook} on this aspect). That being done, we rely on the completeness of $\ol$ to establish the completeness of $\dol$.

We start out by introducing a list of valid formulas, usually called \emph{reduction axioms}, that will be crucial to eliminate dynamic modalities. The next definition follows the presentation of~\cite{DK08}, which also includes postcondition in action models.

\begin{definition}\label{def:redaxioms}
The following equivalences are reduction axioms for eliminating the $[\amodel,e]$
modalities (by applying them on the inner occurrence of a dynamic modality):
\begin{center}
\begin{tabular}{l@{$~\leftrightarrow~$}l}
    $\vdash [\amodel,e] p$ & $(\pre(e) \ra \post(e)(p))$ \\
    $\vdash[\amodel,e]\neg\varphi$ & $(\pre(e) \ra \neg [\amodel,e] \varphi)$ \\
    $\vdash[\amodel,e] (\varphi \wedge \psi)$ & $ ([\amodel,e]\varphi \wedge [\amodel,e] \psi)$ \\
    $\vdash[\amodel,e]\B_a\varphi$ & $(\pre(e)  \ra \bigwedge\set{\B_a[\amodel,f]\varphi \mid e\ra_a f})$ \\
\end{tabular}
\end{center}

\end{definition}

Then, the next lemma (already established in~\cite[Section~2.4]{DK08}) follows. The proof is similar to the proof for classical action models (see e.g.,~\cite[Section~6.6]{DELbook}).

\begin{lemma} \label{lemma:reductionvalid}
The reduction axioms from Definition~\ref{def:redaxioms} are valid.
\end{lemma}

Before going further, let us discuss the properties of observational epistemic models in presence of an update via dynamic modalities. The next lemma states that the class of all observational epistemic models (i.e., those that are Euclidean, transitive and serial) is not closed under the execution of actions. However, if we give up seriality, the system is complete with respect to the obtained class. 

\begin{lemma}\label{lemma:preservation-updates}
Let $\model,w$ be a pointed epistemic observational model, and let $\amodel,e$ be an action model.  Then, 
    \begin{enumerate}
            \item $(\model,w)$ Euclidean, transitive and serial does not imply that $((\model\otimes\amodel),(w,e))$ is serial.  
        \item $(\model,w)$  Euclidean and transitive implies that $((\model\otimes\amodel),(w,e))$ is also Euclidean and  transitive.
    \end{enumerate}
\end{lemma}

Thus, as mentioned, an alternative to obtain completeness consists in relaxing the conditions imposed on the models, concerning the seriality condition in particular. 

\begin{definition}\label{def:dol}
    The axiom system $\dolaxiom$ is defined as $\olaxiom$ minus the axiom $\axm{D}_\B$, and extended with the reduction axioms from Definition~\ref{def:redaxioms}, and the following rule for replacement of equivalents $\axm{RE}$:
    \[
       \infer[\mathsf{(RE)}]{\vdash  \varphi(p/\psi) \leftrightarrow \varphi(p/\psi')}{\vdash\psi \leftrightarrow \psi'}
    \]
    where $\varphi(p/\psi)$ is the result of replacing every occurrence of $p$ in $\varphi$ by $\psi$.
\end{definition}

In order to eliminate all the occurrences of a dynamic modality in a formula (via reduction axioms), some form of replacement of equivalent formulas is needed. In this way, it is possible to apply the equivalences from Definition~\ref{def:redaxioms} as a rewriting system, pushing dynamic modalities inside of the formula and finally eliminating them. This is accomplished by the rule $\axm{RE}$.

\begin{theorem}\label{th:dyncomplete}
The system $\dolaxiom$ is sound and strongly complete with respect to the class of relational models that are Euclidean and transitive.
\end{theorem}

\begin{proof}
First, thanks to Lemma~\ref{lemma:reductionvalid} (using the rule $\axm{RE}$) we get that every formula of $\dol$ is provable (in $\dolaxiom$) equivalent to one of $\ol$. Thus, by Proposition~\ref{th:completeness} and Lemma~\ref{lemma:preservation-updates}, completeness over Euclidean and transitive relational models follows.
\end{proof}

Based on the technical issues described above (i.e., the loss of seriality after updating a model), we obtain some undesired conceptual behaviour. If some agent $a$ believes that $\varphi$ in a serial model (i.e., in which the current state has at least a successor via $a$), naturally it is not the case that the agent also believes $\neg\varphi$. This is due to the fact that some successor of the current point exists, and since all successors must 
satisfy $\varphi$, it cannot be the case that they also satisfy $\neg\varphi$. However, after updating a model via an action that does not preserve seriality (i.e., that removes all successors of the current state), the agent believes both $\varphi$ and $\neg\varphi$, as not having successors trivializes what the agent believes. Therefore, the agent's beliefs may become \emph{inconsistent}. Although this is an undesirable issue, it is a standard characteristic in dynamic epistemic logic, and thus we will not discuss this limitation any further here and choose to proceed with our framework despite it. 


\section{Misdirection via Action Types}
\label{sec:types}

We introduce a simple fragment of $\dol$ that is enough to characterize the intended actions of misdirection, verbal and visual. Based on the notion of \emph{action types} from~\cite{Liu&Wang2013}, this fragment allows only four kinds of \emph{action types} in dynamic modalities, that are either verbal or visual, and either truthful or untruthful in both cases.\footnote{More colloquial terms for truthful and untruthful \emph{visual} actions would be “genuine” and “bogus” visual actions.} 

\begin{remark}\label{remark:pictures}
Both observational epistemic models and action models are labelled directed graphs, with nodes representing states, or worlds, in the former and atomic actions in the latter. To model the distinction graphically, we use circles to represent worlds of observational epistemic models, and squares to represent atomic actions of action models (as in Figures~\ref{fig:modtypestell} and~\ref{fig:modtypesshow}). We use double lines to point out the actual world, both in the observational epistemic models and in action models. 

In observational epistemic models, when an atomic symbol $p$ (from $\prop\cup\obsym$) appears inside a circle representing a world $w$, it indicates that $p$ holds at $w$. If the symbol $p$ is not inside of a world $v$, it means that such a symbol does not hold at $v$. 

In action models, while preconditions are explicitly displayed, postconditions are expressed as the conjuntion of the atomic symbols that are updated, meaning that the symbols become true if they appear in the affirmative form, and become false if they appear in the negated form. For the rest of the symbols that do not appear explicitly, the postcondition acts as the identity function. 
\end{remark}

\subsection{Defining Action Types}
\label{subsec:actiontypes}

Let us define the language $\tslogic$ as the fragment of $\dol$ where dynamic modalities are 
restricted to one of the following action types: {\bf (1)} $\tella^+_a(\varphi)$, {\bf (2)} $\tella^-_a(\varphi)$, {\bf (3)}~$\showa^+_a(\psi)$, or 
{\bf (4)}  $\showa^-_a(\psi)$, with $\varphi$ an arbitrary $\tslogic$-formula, and $\psi=\ell_1 \wedge\dots\ell_n$, where each $\ell_i \in\lit$, and $a\in\agt$. Notice that for visual actions we only admit the action of showing a conjunction of literals. In what follows we say that $\ell\in\psi$, if $\ell=\ell_i$, for some $1\leq i \leq n$. On the one hand, action type {\bf (1)} indicates that agent $a$ truthfully announces formula $\varphi$, whereas {\bf (2)} indicates that $a$ untruthfully announces that $\varphi$. 
On the other hand, in {\bf (3)}, agent $a$ executes a genuine visual action which aims addressees to observe every $\ell\in \psi$ when every $\ell$ actually holds, whereas in {\bf (4)}, agent $a$ executes a bogus visual action which aims addressees to observe that every $\ell\in\psi$ holds when none of them actually hold. As one can  see, definitions of verbal and visual actions are not entirely symmetric since each has its own particularities. For instance, they differ regarding the kind of formulas the action applies to, with verbal announcements applying to arbitrary formulas and visual actions applying only to (conjunctions of) propositional literals. This is due to the fact that our framework focuses on propositional observation via the symbols in $\obsym$. We introduce the formal definitions below.

\begin{definition}[Verbal action types]
    \label{def:verbal-types}
    Let $\varphi$ be a $\dol$ formula, $a\in\agt$ and let $B=\agt\setminus\{a\}$. The action type $\tella^+_a(\varphi)= (\amodel,e)$, where $\amodel=\tup{E,\{\ra_a\}_{a\in\agt},\pre,\post}$, is defined by:
    \begin{itemize}
        \item $E:=\{e,f\}$,
        \item $\ra_a:= E\times E$,
        \item $\ra_b := \{(e,f),(f,f)\}$, for all $b\in B$,
        \item $\pre(e) := \B_a\varphi$, and $\pre(f):= \varphi \wedge \B_a\varphi$, and
        \item $\post(g)(p) := p$, for all $g\in E$ and $p\in \prop\cup\obsym$.
    \end{itemize}
    The action type $\tella^-_a(\varphi)$ is defined exactly as $\tella^+_a(\varphi)$, except that: 
    \begin{itemize}
        \item $\pre(e) := \B_a\neg\varphi$.
    \end{itemize}
\end{definition}

Each of the action types in Definitions~\ref{def:verbal-types} and~\ref{def:visual-types} is a pointed action model $(\amodel,e)$ such that some agent $a\in\agt$ performs an informational action to the other agents (with $e$ being some of its actions). 
Verbal actions update the accessibility relation of agent(s)~$b$, more precisely, it restricts $b$'s accessibilities to those states where formula $\varphi$ holds, while maintaining agent's $a$ accessibilities, and thus her beliefs, as they were. The action type $\tella^+(\varphi)$ corresponds to truthful announcement and is fairly standard, i.e., $a$ believes that $\varphi$ is the case and $a$'s truthful announcement that $\varphi$ restricts $b$'s accessibilities to states where~$\varphi$ holds, such that $b$ also believes that $\varphi$. Its untruthful counterpart $\tella^-(\varphi)$ (inspired by e.g.~\cite{Vanditmarsch2014}) differs from  $\tella^+(\varphi)$ in that $b$'s accessibilities are still restricted to states where $\varphi$ holds but now agent~$a$ (a.k.a., the performer) believes that $\neg\varphi$ holds, contrary to what she announces to $b$. Let us now consider the action types corresponding to visual events. 

\begin{definition}[Visual action types]
\label{def:visual-types}
    Let $\psi=\ell_1\wedge\ldots\wedge\ell_n$ be a conjunction of literals, $a\in\agt$ and let $B=\agt\setminus\{a\}$. The action type $\showa^+_a(\psi)=(\amodel,e)$, where $\amodel=\tup{E,\{\ra_a\}_{a\in\agt},\pre,\post}$, is defined by:
    \begin{itemize}
        \item $E:=\{e,f\}$,
        \item $\ra_a:= E\times E$,
        \item $\ra_b := \{(e,f),(f,f)\}$, for all $b\in B$,
        \item $\pre(e) := \bigwedge_{\ell\in\psi}(\ell \wedge \obs{a}{\ell})$, and $\pre(f):= \bigwedge_{\ell\in\psi} \neg\obs{a}\neg\ell$, and
    \end{itemize}
   \quad  for all $g\in E$,
    \begin{itemize}
        \item $\post(g)(p) := \top$, for all $p\in \{\ell,\obs{b}\ell \mid \ell\in\psi \}$, 
        \item $\post(g)(q) := \bot$, for all $q\in \{\obs{b}\neg\ell \mid \ell\in\psi \}$, and 
        \item $\post(g)(r) := r$, for all $r\in(\prop\cup\obsym)\setminus \{\ell,\obs{b}\ell,\obs{b}\neg\ell \mid \ell \in \psi \}$.
    \end{itemize}
The action type $\showa^-_a(\psi)$ is defined exactly as $\showa^+_a(\psi)$, except that: 
    \begin{itemize}
         \item $\ra_a := \{(e,e),(f,f), (f,e)\}$, 
        \item $\ra_b := \{(e,f), (f,f)\}$, for all $b\in B$.
        \item $\pre(e) := \bigwedge_{\neg\ell\in\psi}(\ell \wedge \obs{a}{\neg\ell})$, and $\pre(f):= \bigwedge_{\ell\in\psi} \neg\obs{b}\neg\ell$, and
        \item $\post(e)(p) := \top$, for all $p\in \{\obs{b}\ell \mid \ell\in\psi \}$,
        \item $\post(e)(q) := \bot$, for all $q\in \{\ell, \obs{b}\neg\ell \mid \ell\in\psi \}$, and 
        \item $\post(e)(r) := r$, for all $r\in(\prop\cup\obsym)\setminus \{\ell,\obs{b}\ell,\obs{b}\neg\ell \mid \ell \in \psi \}$.
        \item $\post(f)(p) := \top$, for all $p\in \{\ell,\obs{a}\ell,\obs{b}\ell \mid \ell\in\psi \}$,
        \item $\post(f)(q) := \bot$, for all $q\in \{\obs{b}\neg\ell, \obs{a}\neg\ell \mid \ell\in\psi \}$, and 
        \item $\post(f)(r) := r$, for all $r\in(\prop\cup\obsym)\setminus \{\ell,\obs{b}\ell,\obs{b}\neg\ell,\obs{a}\ell, \obs{a}\neg\ell \mid \ell \in \psi \}$.
    \end{itemize}
\end{definition}

The first distinction between verbal and visual actions is the use of postconditions for visual actions (also known as `effects', related to the public assignments of~\cite{DitmarschHL12}). This is due to the fact that observations are handled atomically, so in order to update visual perception of agents, we have to update the truth values of the corresponding atomic symbols. The underlying action model in $\showa^+_a(\psi)$ contains two actions. The ``actual'' action to be executed (named $e$) requires as a precondition that all the literals in $\psi$ be true, and leads to update $b$'s observations (for all $b\in B$) to observe the correct facts stated by $\psi$, and to not observe the opposite facts. The latter helps us to stay in the proper class of models after the update. The other possible action ($f$) only requires $a$ to not observe the opposite of $\psi$, and updates $b$'s observations as in~$e$. It is worth noting that the accessibility relation in the announcement says that~$a$ may be uncertain about her own observations, while for $b$ it is established that $a$ does not observe the opposite to what she announces (which is enough for truthfully observing~$\psi$). 

The case of $\showa_a^-(\psi)$ also deserves some explanation. The precondition of the actual action $e$ now states that $\psi$ is entirely false, and that agent $a$ observes that $\psi$ is false. Its postcondition updates $b$'s observation, leading her to wrongly observe $\psi$ as being true. In that respect, action $f$ only requires $b$ not to be able to observe the actual falsity of $\psi$, and to update her observations with $\psi$, making them wrong. Notice that those precondition and postcondition do not involve any doxastic notion of belief or uncertainty: accurate versus inaccurate notions of observation involved are strictly propositional, in the sense of Definition~\ref{def:syntax}. Moreover, while in  $\showa^+_a(\psi)$ there is an edge labelled with~$a$ between $e$ and $f$, this is not the case in  $\showa^-_a(\psi)$. The reason is that in the former, $\pre(f)$ involves $a$'s observations (as $\psi$ is a truthful announcement) while in the latter does not. 
%

\begin{figure}[t]
\begin{tabular}{l@{\quad\quad\quad\quad}llllll}
& & & & &   \begin{tikzpicture}[->]
            \node [] (t) {$\tella^+_a(\varphi)$}; 
            \node [below of = t] { };  
    \end{tikzpicture} 
&
    \begin{tikzpicture}[->]
        \node [state, double, label = {[label-state]left:$e$}] (w1) {$\pre=\B_a\varphi$};      
        \node [state, right = of w1, label = {[label-state]right:$f$}] (w2)
        {$\pre=\varphi\wedge\B_a\varphi$};
        \path (w1) edge[bend left] node [label-edge, above] {\footnotesize{$a,B$}} (w2);
        \path [->] (w2) edge [loop above] node {\footnotesize{$a,B$}} (w2);
        \path (w2) edge[bend left] node [label-edge, below] {\footnotesize{$a$}} (w1);
        \path [->] (w1) edge [loop above] node {\footnotesize{$a$}} (w1);
    \end{tikzpicture}
    \\
    
& & & & &  \begin{tikzpicture}[->]
            \node [] (t) {$\tella^-_a(\varphi)$}; 
            \node [below of = t] { };  
    \end{tikzpicture} 
  &  \begin{tikzpicture}[->]
        \node [state, double, label = {[label-state]left:$e$}] (w1) {$\pre=\B_a\neg\varphi$};  
        \node [state, right = of w1, label = {[label-state]right:$f$}] (w2) {
        $\pre=\varphi\wedge \B_a\varphi$
        };
        \path (w1) edge[bend left] node [label-edge, above] {\footnotesize{$a,B$}} (w2);
        \path [->] (w2) edge [loop above] node {\footnotesize{$a,B$}} (w2);
        \path (w2) edge[bend left] node [label-edge, below] {\footnotesize{$a$}} (w1);
        \path [->] (w1) edge [loop above] node {\footnotesize{$a$}} (w1);
    \end{tikzpicture}
     
\end{tabular} 
\caption{Types of Action Models Based on Distinct Modal Preconditions.} 
\label{fig:modtypestell}
\end{figure}



 

\begin{center}

\begin{figure}[t]
\begin{tabular}{l}
        \begin{tikzpicture}[->]
        \node [state, double, label = {[label-state]left:$e$}] (w1) {
        \begin{tabular}{l@{\,\,}l}
           $\pre=$  & $\bigwedge_{\ell\in \psi} \ell \wedge \obs{a}{\ell}$  \\
           \hline
           $\post=$ &  $\bigwedge_{\ell \in \psi} \ell \wedge \obs{b}{\ell} \wedge \neg\obs{b}{\neg\ell}$ 
         \end{tabular}
        };      
        \node [state, right = of w1, label = {[label-state]right:$f$}] (w2) {
         \begin{tabular}{l@{\,\,}l}
            $\pre=$ & $\bigwedge_{\ell \in \psi}\neg\obs{a}{\neg\ell}$ \\ 
            \hline 
            $\post=$ & $\bigwedge_{\ell \in \psi} \ell \wedge \obs{b}{\ell} \wedge \neg\obs{b}{\neg\ell}$
         \end{tabular}
         };
         \node [left =0.5cm of w1] (t) {$\showa^+_a(\psi)$};
        \path (w1) edge[bend left] node [label-edge, above] {\footnotesize{$a,B$}} (w2);
        \path [->] (w2) edge [loop above] node {\footnotesize{$a,B$}} (w2);
        \path (w2) edge[bend left] node [label-edge, below] {\footnotesize{$a$}} (w1);
        \path [->] (w1) edge [loop above] node {\footnotesize{$a$}} (w1);
    \end{tikzpicture}
    \\
    \begin{tikzpicture}[->]
        \node [state, double, label = {[label-state]left:$e$}] (w1) {         \begin{tabular}{l@{\,\,}l}
            $\pre=$ & $\bigwedge_{\ell \in \psi} \neg \ell \wedge \obs{a}{\neg \ell}$ \\
            \hline 
            $\post=$ & $\bigwedge_{\ell \in \psi} \neg \ell \wedge \obs{b}{\ell} \wedge \neg\obs{b}{\neg\ell}$
         \end{tabular}
         };
         \node [left =0.5cm of w1] (t) {$\showa^-_a(\psi)$}; 
        \node [state, right = of w1, label = {[label-state]right:$f$}] (w2) {
        \begin{tabular}{l@{\,\,}l}
            $\pre=$ & $\bigwedge_{\ell \in \psi}\neg\obs{b}{\neg\ell}$ \\
            \hline 
            $\post=$ & $\bigwedge_{\ell \in \psi} \ell \wedge \obs{a}{\ell} \wedge \obs{b}{\ell}$ \\
                     & \quad \quad $\wedge \neg\obs{a}{\neg\ell} \wedge \neg \obs{b}{\neg\ell}$
         \end{tabular}
        };
        \path (w1) edge[bend left] node [label-edge, above] {\footnotesize{$B$}} (w2);
        \path [->] (w2) edge [loop above] node {\footnotesize{$a,B$}} (w2);
        \path (w2) edge[bend left] node [label-edge, below] {\footnotesize{$a$}} (w1);
        \path [->] (w1) edge [loop above] node {\footnotesize{$a$}} (w1);
    \end{tikzpicture}
     
\end{tabular} 
\caption{Types of Action Models Based on Distinct Observation-Based Preconditions.} 
\label{fig:modtypesshow}
\end{figure}
\end{center}

The action types just defined are displayed in Figures~\ref{fig:modtypestell} and~\ref{fig:modtypesshow}. In what follows, let $a\in\agt$ and denote $B=\agt\setminus\{a\}$. We consider $\psi$ to be a conjunction of literals and $\varphi$ to be a $\tslogic$-formula. For simplicity's sake, we write only $\obsbp$ and $\obsbnp$ to indicate that the proposition $p$, respectively $\neg p$, holds for every $b\in B$. 
In verbal action types (see Figure~\ref{fig:modtypestell}), postconditions are the identity function, i.e., the action does not change the truth value of propositional symbols. In visual action types (see Figure~\ref{fig:modtypesshow}), we update the value of propositional symbols involved in the visual action, as observation is handled at the level of literals (both from $\prop$ and $\obsym$). 
Example~\ref{ex:showing} illustrates the effects of the action types we defined.

\begin{example}\label{ex:showing}
Consider the  observational epistemic model displayed below on the left, in which agent $a$ is certain about observing $p$, but agent $b$ is uncertain whether she observes $p$ or $\neg p$. 
Agent $b$ is also uncertain about what agent $a$ observes, i.e., whether agent $a$ observes $p$ or $\neg p$. Moreover, agent $b$ considers possible to observe $\neg p$. 
The actual situation is the following one in which $p$ holds:

    \begin{center}
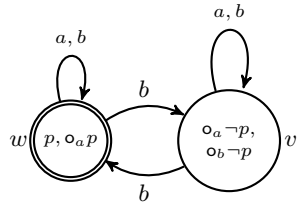
    \begin{tabular}{ll}
    \begin{minipage}{0.4\textwidth}
    \begin{tikzpicture}[->]
        \node [state, double, circle, label = {[label-state]left:$w$}] (w1) {$p,\obs{a}{p}$};    \path [->] (w1) edge [loop above] node {\footnotesize{$a,b$}} (w1);  
        \node [state, circle, right = of w1,label = {[label-state]right:$v$}] (w2) {$\begin{array}{c}\obs{a}{\neg p}, \\ \obs{b}{\neg p} \end{array}$ };
        \path (w1) edge [bend left] node [label-edge, above] {$b$} (w2);
        \path (w2) edge [bend left] node [label-edge, below] {$b$} (w1);
        \path [->] (w2) edge  [loop above] node {\footnotesize{$a,b$}} (w2);
    \end{tikzpicture}
    \end{minipage}
    &
    \begin{minipage}{0.5\textwidth}
            \begin{tikzpicture}[->]
        \node [state, double, circle, label = {[label-state]left:$(w,e)$}] (w1) {$\begin{array}{c} p, \obs{a}{p} \\ \obs{b}{\neg p}\end{array}$};    \path [->] (w1) edge  [loop above] node {\footnotesize{$a$}} (w1);  
        \node [state, circle, right = of w1,label = {[label-state]right:$(w,f)$}] (w2) {$\begin{array}{c}\obs{a}{\neg p} \\ \obs{b}{\neg p}\end{array}$};
        \node [state, circle, below = of w2,label = {[label-state]right:$(v,f)$}] (w3) {$\begin{array}{c} \obs{a}{\neg p} \\ \obs{b}{\neg p}\end{array}$};

        \path (w1) edge [bend left] node [label-edge, above] {$b$} (w2);                \path (w1) edge [bend right] node [label-edge, above] {$b$} (w3);
        \path (w2) edge [bend left] node [label-edge, above] {$a$} (w1);
        \path (w2) edge [bend left] node [label-edge, right] {$b$} (w3);                \path (w3) edge [bend left] node [label-edge, left] {$b$} (w2);

        \path [->] (w2) edge [loop above] node {\footnotesize{$a$}} (w2);
        \path [->] (w3) edge [loop below] node {\footnotesize{$a,b$}} (w23);
    \end{tikzpicture}
    \end{minipage}
    \end{tabular}
    \end{center}
Now suppose that the dynamic action $\showa^-_{a}(\neg p)$ is executed, i.e., that agent $a$ misdirects agent~$b$'s visual perception, leading her to wrongly observe that $\neg p$ is happening. The resulting observational epistemic model is the one on the right in the figure above.
    Thus, we can conclude that:
    \[
    \model,w\models\obs{a}{p}\wedge\neg\obs{b}{\neg p} \wedge [\showa^-_a(\neg p)](p\wedge \obs{b}{\neg p}\wedge \B_b \obs{a}{\neg p}). 
    \]
    
    Intuitively, this means that agent $a$ successfully misdirects agent's $b$ attention to~$\neg p$: after $a$ performs the action of showing $\neg p$, agent $b$ (wrongly) observes $\neg p$ and believes that agent~$a$ also observes $\neg p$. Notice that we can also replace $[\showa^-_a(\neg p)]$ by $\tup{\showa^-_a(\neg p)}$, as the precondition of the action holds at the current state. If we want to force the action to be executed only in this case, the form $\tup{\showa^-_a(\neg p)}$ should be used.
\end{example}

\begin{remark} In this paper, we focus on what the spectators believe and observe about the magician's beliefs and observations. So far, we have not discussed the status of other second-order beliefs, such as what one spectator believes about another spectator's beliefs. Note that $\tella$ actions are, in fact, a particular form of semi-private announcements (see, e.g.,~\cite[Chapter 6]{DELbook}); therefore, everything announced to a set of agents (here, the spectators) becomes common knowledge.\footnote{We use common knowledge' because it is the standard terminology in Epistemic Logic, although, strictly speaking, we are concerned with beliefs here.} Hence, after the $\tella$ actions, if $\B_b\B_a\varphi$ (with $b$ a spectator and $a$ the magician), then $\B_c\B_b\B_a\varphi$ (where $c$ is also a spectator).
\end{remark}

In Section~\ref{sec:other}, we will discuss how the previously defined action types enable us to characterize concepts such as simulation of the false and dissimulation of the truth. Moreover, we will discuss the logical implications of such formalizations.

\subsection{Misdirection in the French Drop trick}
\label{subsec:example}

With the machinery presented above, we can now formalize the French Drop trick as an illustration of the dynamic logic of misdirection $\dol$. Informally, the French Drop trick can be decomposed into three distinct steps:

\begin{enumerate}

\item First, there is the \emph{\textbf{opening}} of the trick: the magician shows a coin to the public, shakes her hands and states that the coin is in one of her two hands, right or left. 

\item Second, there is the \emph{\textbf{method}} used by the magician, namely the actions she performs in front of the audience, and how those actions result in misdirection. On the one hand, there is an action of \emph{simulation} through a fake passing of the coin. The magician seems to move the coin to her right hand but in fact the coin stays in her left hand. On the other hand, at the same time, there is an action of \emph{dissimulation} through a genuine grabbing. The magician conceals the coin in her left hand by palming it. In the current section, we only model the action of simulation that takes place during the French Drop since, as we will show in subsection \ref{subsubsec:derived}, dissimulation is a derived notion defined as the absence of simulation, verbal or visual. That being said, the magician can reinforce her effect by also making verbal announcements when passing and palming the coin, though not necessarily.

\item Finally, there is an \emph{\textbf{effect}} produced by actions of simulation (\textit{fake passing}) and dissimulation (\textit{genuine palming}). The magician opens her right hand, and thus shows to the public that the coin is not in there, to the great surprise of the crowd.    
\end{enumerate}

 Based on this informal description, we now model the French Drop trick in the fragment $\tslogic$. The three steps of the magic trick (opening, method and effect) are called the initial state, the epistemic action and the posterior state in the modelling herein.   

\begin{itemize}
\item \textbf{Initial state:} we have two propositional symbols: $l$ that stands for \emph{“the coin is in the magician's left hand''}, and $r$ for \emph{“the coin is in the magician's right hand''}. Accordingly, we have four possible states, but the state in which the coin is at the same time in the magician's left and right hand is not actually plausible.  
Precisely, we have three options: either the coin is in the left hand ($l \wedge \neg r$), the coin is in the right hand ($\neg l \wedge r$), or there is no coin in the magician's hands ($\neg l \wedge \neg r$). The actual state is the one in which the coin is in the magician's left hand: ($l \wedge \neg r$). The following pointed observational epistemic model $\model,w$ models this opening situation:%
    \begin{center}
    \begin{tikzpicture}[->]
        \node [state, double, circle, label = {[label-state]below:\footnotesize$w$}] (w1) {$\begin{array}{c} l, \obs{a}{l}, \\ \obs{a}{\neg r}\end{array}$};    \path [->] (w1) edge  [loop left] node {\footnotesize{$a,b$}} (w1);  
        \node [state, circle,  above right = 2.5em and 5em of w1, label = {[label-state]above:\footnotesize$v$}] (w2) {$\begin{array}{c} r, \obs{a}{r} \\ \obs{a}{\neg l}\end{array}$};
        \node [state, circle,  below right = 2.5em and 5em of w1, label = {[label-state]below:\footnotesize$u$}] (w3) {$\begin{array}{c} \obs{a}{\neg r}, \\ \obs{a}{\neg l}\end{array}$};

        \path (w1) edge [bend left] node [label-edge, above] {\footnotesize$b$} (w2);
        \path (w1) edge [bend right] node [label-edge, above] {\footnotesize$b$} (w3);
        \path (w2) edge [bend right,in=198] node [label-edge, left] {\footnotesize$b$} (w3);
        \path (w3) edge [bend right] node [label-edge, right] {\footnotesize$b$} (w2);

        \path (w3) edge [bend right] node [label-edge, below] {\footnotesize$b$} (w1);
        \path (w2) edge [bend left,out=380] node [label-edge, above] {\footnotesize$b$} (w1);
        
        \path [->] (w2) edge [loop right] node {\footnotesize{$a,b$}} (w2);
        \path [->] (w3) edge  [loop right] node {\footnotesize{$a,b$}} (w3);

    \end{tikzpicture}
    \end{center}%
Initially, only the magician believes that the coin is in her left hand. Moreover, on each possible situation, the magician, here represented by agent $a$, observes whether the coin is in a given hand and not in the other hand. For instance, in the actual state~$w$, she observes that the coin is in her left hand, and that it is not in her right hand:~($l \wedge \neg r$). But the audience is uncertain whether the coin is in the magician's left or right hand, and has not visual information. Let $a$ be the agent modeling the magician, and w.l.o.g., let us assume there is only one agent $b$ in the audience.

\item \textbf{Visual actions:} we know that the magic actions are two-sided in the French Drop: on the one hand, it consists in pretending that $r$ happens while, on the other hand, it consists in hiding that $l$ remains the case. As indicated, we only model simulation (\textit{fake passing}) here. This corresponds to action type $\showa^-_a(r\wedge \neg l)$ and is illustrated below. With such an action, the magician aims at making the audience believe that the event $r$ is taking place, when in fact, the magician herself believes and observes that it is not the case ($\B_a\neg r \wedge \obs{a}\neg r$). Also, it is obvious that the magician clearly observes and believes that the coin is still in her left hand ($\B_a l \wedge \obs{a}l$), as she is observing the real situation.
    
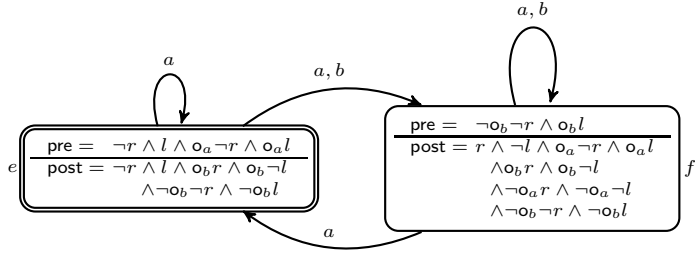
\begin{center}
    \begin{tikzpicture}[->]
        \node [state, double, label = {[label-state]left:\footnotesize$e$}] (w1) {         \begin{tabular}{l@{\,\,}l}
            $\pre=$ & $ \neg r \wedge l \wedge \obs{a}{\neg r} \wedge \obs{a}{l}$ \\
            \hline 
            $\post=$ & $ \neg r \wedge l \wedge \obs{b}{r} \wedge \obs{b}{\neg l}$ \\
                     & \quad \ $\wedge \neg\obs{b}{\neg r} \wedge \neg \obs{b}{l}$
         \end{tabular}
         };
        \node [state, right = of w1, label = {[label-state]right:\footnotesize$f$}] (w2) {
        \begin{tabular}{l@{\,\,}l}
            $\pre=$ & $\neg\obs{b}{\neg r} \wedge \neg\obs{b}{l}$ \\
            \hline 
            $\post=$ & \ \ \ \ $r \wedge \neg l$ \\
                    & \ \ $\wedge\obs{a}{r}\wedge \obs{a}{\neg l}\wedge \obs{b}{r} \wedge \obs{b}{\neg l}$ \\
                    &  \ \  $\wedge \neg\obs{a}{r} \wedge \neg \obs{a}{\neg l}\wedge \neg\obs{b}{\neg r} \wedge \neg \obs{b}{l}$
         \end{tabular}
        };
        \path (w1) edge[bend left] node [label-edge, above] {\footnotesize{$b$}} (w2);
        \path [->] (w2) edge [loop above] node {\footnotesize{$a,b$}} (w2);
        \path (w2) edge[bend left] node [label-edge, below] {\footnotesize$a$} (w1);
        \path [->] (w1) edge [loop above] node {\footnotesize{$a$}} (w1);
    \end{tikzpicture}
    \end{center}
    Below we depict the model obtained after executing the action $\showa^-_a(r\wedge\neg l)$ in the opening scenario (here states $(v,e)$ and $(u,e)$ are ignored in the picture since $\model,v\not\models\pre(e)$ and $\model,u\not\models\pre(e)$):
    
    \begin{center}
    \begin{tikzpicture}[->]
        \node [state, double, circle, label = {[label-state]above:\footnotesize$(w,e)$}] (w1) {$\begin{array}{c} l, \obs{a}{l}, \\ \obs{a}{\neg r}, \obs{b}{r} \\ \obs{b}{\neg l} \end{array}$};    
        
         \node [state, circle,  below = 3em of w1, label = {[label-state]below:\footnotesize$(v,f)$}] (w3) 
         {$\begin{array}{c} r, \obs{a}{\neg l}, \obs{a}{r},  \\ \obs{b}{\neg l}, \obs{b}{r}  \end{array}$};

        \node [state, circle,  left  = 3em of w3, label = {[label-state]left:\footnotesize$(w,f)$}] (w2) 
        {$\begin{array}{c} r, \obs{a}{\neg l}, \obs{a}{r},  \\ \obs{b}{\neg l}, \obs{b}{r}  \end{array}$};

          \node [state, circle,  right= 3em of w3, label = {[label-state]right:\footnotesize$(u,f)$}] (w4) 
          {$\begin{array}{c} r, \obs{a}{\neg l}, \obs{a}{r},  \\ \obs{b}{\neg l}, \obs{b}{r}  \end{array}$};
        
         \path (w1) edge [bend right]  node [label-edge, above] {\footnotesize{$b$}} (w2);
         \path (w1) edge [bend right] node [label-edge, left] {\footnotesize{$b$}} (w3);
         \path (w1) edge [bend left] node [label-edge, above] {\footnotesize{$b$}} (w4);
         
         \path (w2) edge [bend left] node [label-edge, above] {\footnotesize{$b$}} (w3);
         \path (w3) edge [bend left] node [label-edge, below] {\footnotesize{$b$}} (w2);
         \path (w3) edge [bend left] node [label-edge, above] {\footnotesize{$b$}} (w4);
         \path (w4) edge [bend left] node [label-edge, below] {\footnotesize{$b$}} (w3);

           \path  [->, every arrow/.style={looseness=10}]  (w1) edge [in=140,out=160,loop] node [label=left: \footnotesize{$a$}] {} (w1);  
      
         \path [->, every arrow/.style={looseness=8}]  (w2) edge [in=-120,out=-60] node [label-edge, below] {\footnotesize{$b$}} (w4);

         \path [->, every /.style={looseness=30}]  (w4) edge [in=-50,out=-130] node [label-edge, above] {\footnotesize{$b$}} (w2);

        \path [->,every loop/.style={looseness=8}] (w2) edge [in=120,out=140,loop] node [label=left: \footnotesize{$a,b$}] {} (w2);
        \path [->,every loop/.style={looseness=8}] (w3) edge [in=50,out=70,loop] node [label=right: \footnotesize{$a,b$}] {} (w3); 
        \path [->,every loop/.style={looseness=8}] (w4) edge [in=50,out=30,loop] node [label=right: \footnotesize{$a,b$}] {} (w4);
      \end{tikzpicture}
    \end{center}

    Notice that after this misdirection by simulation, all the states considered possible by agent $b$ (the audience) become states in which she observes that the coin is in the magician's right hand. Observation became epistemic after the action: agent $b$ also believes this (wrong) fact she is observing.
    
    Finally, the execution of $\showa^+_a(l\wedge \neg r)$ (depicted below) models the following fact: the magician reveals to the audience that they have been misdirected. This corresponds to the visual action of showing that the coin is actually in the magician's left hand, as opposed to what the audience observed previously. 

    \begin{center}
    \begin{tikzpicture}[->]
        \node [state, double, label = {[label-state]left:\footnotesize$e$}] (w1) {         \begin{tabular}{l@{\,\,}l}
            $\pre=$ & $ l \wedge \neg r \wedge   \obs{a}{l}  \wedge \obs{a}{\neg r} $ \\
            \hline 
            $\post=$ & $l \wedge  \neg r \wedge  \obs{b}{l} \wedge \obs{b}{\neg r}$ \\
                     & \ \ $\wedge  \neg \obs{b}{\neg l} \wedge \neg\obs{b}{r}$
         \end{tabular}
         };
        \node [state, right = of w1, label = {[label-state]right:\footnotesize$f$}] (w2) {
        \begin{tabular}{l@{\,\,}l}
            $\pre=$ & $\neg\obs{a}{\neg l} \wedge \neg\obs{a}{r}$ \\
            \hline 
            $\post=$ & $l \wedge \neg r \wedge \obs{b}{l} \wedge \obs{b}{\neg r}$ \\
            & \ \ $\wedge \neg\obs{b}{\neg l} \wedge \neg \obs{b}{r}$
         \end{tabular}
        };
        \path (w1) edge[bend left] node [label-edge, above] {\footnotesize{$a,b$}} (w2);
        \path [->] (w2) edge [loop above] node {\footnotesize{$a,b$}} (w2);
        \path (w2) edge[bend left] node [label-edge, below] {\footnotesize$a$} (w1);
        \path [->] (w1) edge [loop above] node {\footnotesize{$a$}} (w1);
    \end{tikzpicture}
    \end{center}

    The result of executing the truthful visual action $\showa^+_a(l \wedge \neg r)$ is depicted below:
    
    
    
    \begin{center}
    \begin{tikzpicture}[->]
        \node [state, double, circle, label = {[label-state]left:\footnotesize$((w,e),e)$}] (w1) {$\begin{array}{c} l, \obs{a}{l}, \\ \obs{a}{\neg r}, \obs{b}{\neg r} \\ \obs{b}{l} \end{array}$};    
        \node [state, circle,  right  = 5em of w1, label = {[label-state]right:\footnotesize$((w,e),f)$}] (w2) {$\begin{array}{c} l, \obs{a}{l},  \\ \obs{a}{\neg r},  \obs{b}{\neg r}, \\ \obs{b}{l}\end{array}$};



         \path  [->, every arrow/.style={looseness=10}]  (w1) edge [in=140,out=160,loop] node [label=left: \footnotesize{$a$}] {} (w1);   
         \path (w1) edge [bend left]  node [label-edge, above] {\footnotesize{$a$}} (w2);
         \path (w2) edge [bend left] node [label-edge, below] {\footnotesize{$a$}} (w1);
         

        \path [->,every loop/.style={looseness=8}] (w2) edge [in=50,out=30,loop] node [label=right: \footnotesize{$a$}] {} (w2);
      \end{tikzpicture}
    \end{center}
    Again, as a result, agent $b$ epistemically observes (this time, correctly) where the coin actually is, i.e. in the magician's left hand. 
    
    \item \textbf{Posterior state:} naturally, $l$ is still true. 
    Formally, we get:
    \[
    \begin{array}{lll}
    \model,w &\models & l \wedge \obs{a}{l} \wedge \obs{a}{\neg r} \wedge \neg\obs{a}{r} \wedge \neg\obs{b}{r}  \\
    & & \wedge \ \tup{\showa^-_a(r \wedge\neg l)}(\obs{a}{l} \wedge \obs{a}{\neg r} \wedge \obs{b}{r} \wedge \obs{b}{\neg l} \\
    & & \phantom{\tup{\showa^-_a(r\wedge\neg l)}}\  \wedge \ \tup{\showa^+_a(l\wedge \neg r)}(\obs{b}{l} \wedge \obs{b}{\neg r})).
    \end{array}
    \]
\end{itemize}

Each line in the formula above corresponds to one of the moments in which an action has been performed, and to their corresponding effects. In the first line, as described in the initial state, only the magician has a correct picture of the actual situation.
Based on this initial state, as shown in the second line, the action of misdirection by simulation is executed, making $b$ (wrongly) observe that the coin is in the magician's right hand ($\obs{b}{r}$) and no longer in her left hand ($\obs{b}{\neg l}$). Magician's beliefs and observations remain unchanged, as expected. When the trick is revealed, as shown in the third line, both the magician and spectators both have an adequate picture of the coin being in the magician's left hand ($l$) and not in the magician's right hand ($\neg r$). Their factual observations are correct: in the actual state (i.e. $((w,e),e)$), we have that $\obs{a}{l}$ and $\obs{a}{\neg r}$ hold for the magician; $\obs{b}{l}$ and $\obs{b}{\neg r}$ hold for the spectators. Their beliefs are correct too: $\B_a l$ and $\B_a \neg r$ for the magician; $\B_b l$ and $\B_b \neg r$ for the spectators. In addition to that, the magician and spectators have conscious beliefs of their own observations: it is the case that $\B_a(\obs{a}{l})$ and $\B_a(\obs{a}{\neg r})$, and also that $\B_b(\obs{b}{l})$ and $\B_b (\obs{b}{\neg r})$. 

It is worth to notice that every action performed by the magician (in this case, $\showa^-_a(r\wedge\neg l)$ and $\showa^+_a(l \wedge\neg r)$) are executed only whenever their corresponding preconditions hold at their respective current situation, i.e., we use a diamond-like modality. This means that the magician \emph{intentionally} performs these actions, in particular, she intentionally misdirects the spectators with $\showa^-_a(r\wedge\neg l)$. By contrast, if we use a box-like modality (i.e., $[\showa^-_a(r\wedge\neg l)]$), the formula will also hold when the precondition fails to be true; meaning that in case of misdirection, either the magician is not actually misdirecting the spectators, or that she unintentionally misdirects them. We focus herein on intentional actions, as discussed previously.


\section{On the Generality of $\dol$}
\label{sec:other}

This section shows that $\dol$ can be used to characterize more specific and general concepts, based on combining atomic observations, belief modalities and dynamic actions. We first define a stronger notion of belief, by reinforcing the basic beliefs of $\dol$ with dynamic actions, and a stronger notion of observation, by using atomic observations to define epistemic observations. We then discuss how principles involving beliefs and observations interact, and why some of them should not be stated as axioms. Afterwards, we proceed to characterize the difference between misdirection as simulation (which involves an action) and misdirection as dissimulation (which involves the lack of an action). We finally attempt to capture the effect of surprise involved in the French Drop trick. 

\subsection{Stronger notions of Observation and Belief}
\label{sec:strongnotions}

 The way logic $\dol$ is interpreted in Definition~\ref{def:semantics}, the notion of observation, written~$\obs{}{}$, is atomic and non-epistemic. However, using a modal logic framework gives us the possibility to define a notion of epistemic observation, written $\Ob$, based on this initial concept of atomic (i.e., non-epistemic) observation $\obs{}{}$.

\begin{definition}\label{def:epistemic-observation}
We define epistemic observation, written $\Ob$, as follows:
\[ 
 \Ob_b p\ :=\ \obs{b}{p} \wedge \B_b(\obs{b}{p}).
\]
\end{definition}

According to this definition, agent $a$ \textit{epistemically observes} that $p$ when agent $b$ \textit{atomically observes} that $p$, i.e. $\obs{b}{p}$, but also believes that she is atomically observing that $p$, i.e. $\B_b(\obs{b}{p})$. That being done, we define a notion of epistemic observation which is \emph{stronger} in the sense that it associates epistemic observations to dynamic actions. 

\begin{definition}\label{def:strong-epistemic-observation}
We define strong epistemic observation, written $\ES$, as follows:
\[ 
\ES_b p\ := \ \Ob_b p \wedge \bigvee_{a\in\agt\setminus\{b\}}\tup{\tella^+_a(p)}\B_b p \vee \tup{\showa^+_a(p)}\B_b p
\]
\end{definition}

Strong epistemic observation should be understood as epistemic observations that turn out to be supported by verbal and/or visual actions at the same time. 

In contrast with our initial notion of belief~$\B$, another notion we can define is \emph{strong belief}. As it is, the modality $\B$ defines belief as being weak doxastic attitude in which the agent may lack any justification for supporting her beliefs. On the contrary, we define strong beliefs as being weak beliefs supported by genuine actions, either verbal or visual, as in the following definition.   

\begin{definition} Let $\varphi$ be a conjunction of literals, we define strong beliefs, written $\SB$, in the following way:
\[ 
\SB_b\varphi\ := \ \B_b\varphi \wedge \bigvee_{a\in\agt\setminus\{b\}}\tup{\tella^+_a(\varphi)}\B_b\varphi \vee \tup{\showa^+_a(\varphi)}\B_b\varphi
\]
\end{definition}

This definition means that agent $b$ believes that $\varphi$ in a strong sense when $b$ ‘weakly' believes that $\varphi$ and this belief is robust under any truthful action with $\varphi$, verbal ($\tup{\tella^+_a(\varphi)}$) or visual ($\tup{\showa^+_a(\varphi)}$). In other words, agent $b$ holds a strong belief that $\varphi$ because $b$'s belief is supported by the possibility of agent $a$'s action.

\subsection{Interaction between Beliefs and Observations}
\label{subsec:dox-obs-principles}

Let us discuss some principles linking observations and beliefs, that are either satisfiable, valid or unreasonable in our setting. The first principle, hereafter separated into (1a) and (1b), states that if an agent $a$ observes a certain fact as true (respectively, as false), $a$ believes her observations to be true (respectively, false): 

\begin{itemize}
    \item[](1a) $\obsp \rightarrow \B_a \obsp$
    \item[](1b) $\obsnp \rightarrow \B_a \obsnp$
\end{itemize}

Principles (1a)-(1b) are \textit{satisfiable} but not valid in our setting. As indicated before, we aim to keep a distinction between beliefs and observations since observation consists in sensory visual perception that does not necessarily involve conscious representation \textcolor{black}{(for instance, you may see an entity in the water without seeing it \textit{as} a fish in the water).} So, if we accept principles (1a) and (1b) as universally valid, observations and belief conflate into the concept of ``epistemic observation'' discussed in  Definition~\ref{def:epistemic-observation}. Naturally, these statements are satisfiable: it can be the case that an agent observes a fact $p$ (respectively, $\neg p$) and that she believes she is observing $p$ (respectively, $\neg p$).

The following principle, now divided into (2a) and (2b), states that an agent cannot have wrong beliefs about its observations:

\begin{enumerate}
    \item[](2a) $\B_a \obsp \rightarrow \obsp$
    \item[](2b) $\B_a \obsnp \rightarrow \obsnp$
\end{enumerate}

Principles (2a)-(2b) are \textit{satisfiable} since there are states in which an agent has correct beliefs about her observations. That being said, they are not valid in our setting since, otherwise, an agent would never be misdirected, contrary to our purposes. To conclude, consider the following principle, separated into (3a) and (3b):

\begin{enumerate}
    \item[](3a) $(\obsp\wedge \B_a\obsp) \rightarrow [\showa_a^+(p)] (\obs{b}{p}\wedge \B_b \obs{b}{p})$ 
    \item[](3b) $(\obsnp \wedge \B_a\obsnp) \rightarrow [\showa_a^-(p)](\obs{b}{p}\wedge \B_b \obs{b}{p})$
\end{enumerate}

Principles (3a)-(3b) are now \textit{valid} in our framework. Instead of relating beliefs and observations statically, they make use of the dynamic actions in order to guarantee that an observation is believable. Formula (3a) states that whenever agent $a$ epistemically observes that $p$, then, if $a$ truthfully shows that $p$ to some addressee $b$, then $b$ will also come to epistemically observe that $p$ too. Formula (3b) states the same success but with simulation that $p$: whenever agent $a$ epistemically observes that $\neg p$, then, if $a$ now untruthfully shows that $p$ to some addressee $b$, then $b$ will come to epistemically observe that $p$. These principles reflect 
properly the spirit of the $\dol$ framework.


\subsection{Other Variations on Misdirection}















\subsubsection{Misdirection as Simulation}

Having defined agent $a$'s verbal simulation that $\neg p$ with $\tella^-_a(\neg p)$ and agent $a$'s visual simulation that $\neg p$ with $\showa^-_a(\neg p)$, we can define agent $a$'s general simulation that~$\neg p$ by the disjunction of both verbal and visual types, as follows. 

\begin{definition}\label{def:sim} We define simulation of the false, written $\Sim$ for simulation, by the following disjunction:
\begin{align*}
\Sim_a \neg p\ :=\ &  (\B_a p \wedge  \tup{\tella^-_a(\neg p)}\B_b \neg p)  \vee (\obs{a}{p} \wedge \tup{\showa^-_a(\neg p)}\obs{b}{\neg p})
\end{align*}
where $p\in\prop$, and $\Sim_a \neg p$ expresses that \emph{``agent $a$ simulates that $\neg p$ is the case to agent $b$''}.
\end{definition}

Definition~\ref{def:sim} expresses simulation of a factual information $\neg p$ as being either verbal simulation of $\neg p$ (left-hand disjunct with $\tella^-$) or visual simulation of $\neg p$ (right-hand disjunct with $\showa^-$). Since the formula $\Sim_a \neg p$ is a disjunction of cases, it is sufficient that either verbal simulation obtain or that visual simulation obtain for $\Sim_a \neg p$ itself to obtain. In the left-hand disjunct of $\Sim_a \neg p$, verbal simulation is expressed by the conjunction $\B_a p \wedge \tup{\tella^-_a(\neg p)}\B_b \neg p$. This means that agent $a$ verbally simulates that~$\neg p$ to agent $b$ when $a$ believes that $p$ but $a$ tells $b$ that~$\neg p$ such that, as a result of the verbal action $\tella^-_a(\neg p)$, agent $b$ believes that $\neg p$. In the right-hand disjunct of $\Sim_a \neg p$, now, visual simulation is defined by the conjunction $\obs{a}{p} \wedge \tup{\showa^-_a(\neg p)}\obs{b}{\neg p}$. This means that agent $a$ visually simulates that $\neg p$ to agent $b$ when $a$ visually observes that $p$ and, as a result of the visual action $\showa^-_a(\neg p)$, agent $b$ observes that~$\neg p$. In both verbal and visual simulation of $\neg p$, the aim of $a$ is that $b$ holds a wrong visual representation that $\neg p$ through the action performed. Moreover, we use again diamond-like dynamic modalities since we aim to model the actions performed by the magician are intentional.

\subsubsection{Misdirection as Dissimulation}
\label{subsubsec:derived}


So far, we have used $\dol$ to express the notion of simulation, or faking, in the verbal case ($\tella^-$) and in the visual case ($\showa^-$). The operator, e.g. the magician, uses a verbal or visual action to direct the addressees' attention, e.g. the spectators', toward a false representation of reality. As noticed in Section~\ref{sec:concepts}, however, this kind of misdirection based on an effective action differs from any other kind of misdirection, known as \textit{``dissimulation''}, or hiding, in which no action takes place in order to deceive. In this case, deception \emph{does not} result from any action, but from the failure of a verbal or visual action. Quite fortunately, $\dol$ is expressive enough to also capture such a form of information dissimulation consisting in withholding information.

\begin{definition}\label{def:omit} We define dissimulation of information, written $\Dis$, by the following disjunction:
\begin{align*}
\Dis_a p\ :=\ &  (\B_a p \wedge \B_b\neg p \wedge \neg\tup{\tella^+_a(p)}\B_b \neg p)   \vee (\obs{a}{p} \wedge \obs{b}{\neg p} \wedge \neg\tup{\showa^+_a(p)}\obs{b}{\neg p})
\end{align*}
where $p\in\prop$, and $\Dis_a p$ expresses that \emph{``agent $a$ dissimulates that $p$ is the case to agent $b$''}.
\end{definition}

Definition~\ref{def:omit} expresses dissimulation of a factual information $p$ as being either verbal dissimulation of $p$ (left-hand disjunct referring to $\tella^+$) or visual dissimulation of $p$ (right-hand disjunct referring to $\showa^+$). Since the formula $\Dis_a p$ is a disjunction of cases, it is sufficient that either verbal dissimulation obtain or that visual dissimulation obtain for $\Dis_a p$ itself to obtain. In the left-hand disjunct of $\Dis_a p$, verbal dissimulation is expressed by the conjunction $\B_a p \wedge \B_b\neg p \wedge \neg\tup{\tella^+_a(p)}\B_b \neg p$. Here, this means that agent $a$ verbally hides that $p$ to agent $b$ when $a$ believes that $p$ and $b$ believes that $\neg p$. But crucially, contrary to Definition \ref{def:sim} above, now \textit{no} \textit{verbal action} $\tella^+_a(p)$ has participated to make agent $b$ actually believe that $\neg p$, i.e. not believe that $p$. In other words, agent $b$ believes that $\neg p$ \textit{by the lack of an action}. In the right-hand disjunct of $\Dis_a p$, now, visual dissimulation is defined by the conjunction $\obs{a}{p} \wedge \obs{b}{\neg p} \wedge  \neg\tup{\showa^+_a(p)}\obs{b}{\neg p}$. This means that agent $a$ visually hides that $p$ to agent $b$ when $a$ visually observes that $p$, $b$ observes that $\neg p$ and here again, crucially, \textit{no} \textit{visual action} $\showa^+_a(p)$ has participated to make agent $b$ observe $\neg p$, i.e., not observe $p$. In both verbal and visual dissimulation of $p$, the aim of $a$ is that $b$ hold the wrong representation that $\neg p$, not as the result of any action but, specifically, by the lack of an action.

%

Recall that in the French Drop trick, dissimulation is not omission of information, or withholding information, but hiding information. The magician performs some visual action, i.e. some hand gesture, to make the audience wrongly observe that the coin moved to her right hand. In our own variant of the trick, however, the magician finally reveals her secret at the end of the trick: she opens her left hand in front of the public. But if the magician had not revealed her secret, this would have counted as a case of dissimulation in the sense of withholding, or omitting, information. In doing so, more precisely in doing nothing to correct the spectators' mistaken observations, the magician would have omitted visual information from them. The strength of logic $\dol$ is to be able to tease apart dissimulation via hiding from dissimulation via omission.

\subsubsection{Simulation with Dissimulation}
\label{subsubsec:imp}

Simulation of the false is a notion conceptually distinct from dissimulation of the truth (see subsection \ref{subsec:twomain}). We have used $\dol$ to make this difference explicit by distinguishing simulation ($\Sim$) from dissimulation ($\Dis$). But though simulation and dissimulation are conceptually distinct, they are also logically related.   

Intuitively, an agent who makes it appear that $\neg p$ (simulation) thereby also conceals the true state of affairs that $p$ (dissimulation). Accordingly, in logic $\dol$, if agent $a$ simulates the falsity of $p$, either verbally by performing the action $\tella^-_a(\neg p)$ or visually by performing the action $\showa^-_a(\neg p)$, agent $a$ also comes to dissimulate the truth of $p$ at the same time. This is the case in the French Drop trick when the magician performs the visual action of simulating that $(\neg r \wedge l)$ (i.e. that the coin is not in the magician's right hand but in the magician's left hand): $\showa^-_a(r\wedge\neg l)$. Naturally, at the same time, the magician comes to not show, i.e. to dissimulate, the truth of $(\neg r \wedge l)$: no execution of $\showa^+_a(\neg r \wedge l)$ reveals such value (see subsection~\ref{subsec:example}).

\begin{proposition}\label{prop:simulation}
 Let $p\in\prop$ and let $\model,w$ be a pointed model, then:
\begin{enumerate}
\item \textbf{Verbal case:} $\model,w \models \tup{\tella^-_a(\neg p)}\B_b \neg p \rightarrow \neg \tup{\tella^+_a(p)}(\Bdiam_b \top \wedge \B_b \neg p)$;
\item \textbf{Visual case:} $\model,w \models \tup{\showa^-_a(\neg p)}\obs{b}{\neg p}  \rightarrow \neg \tup{\showa^+_a(p)}\obs{b}{\neg p}$.
\end{enumerate}
\end{proposition}

\begin{proof}
1) For the verbal case, let us suppose that $\model,w \models \tup{\tella^-_a(\neg p)}\B_b \neg p$. By the definition of $\models$, we have that: (a) $\model,w \models\B_a p$, and that (b) $\model',(w,e)\models \B_b\neg p$, where $\model'= \model \otimes \tella^-_a(\neg p)$.     
Then, for agent $b$ (the spectator) the only possible world after the update, is one in which the action $f$ from $\tella^-_a(\neg p)$ was executed, i.e., a world that initially satisfied formula  ($\neg p\wedge\B_a\neg p$). Thus, for every possible successor of $(w,e)$ via agent~$b$, formula $\neg p$ holds (notice that since $p$ is atomic, we can be sure there is no Moorean phenomena there\footnote{For a more extensive discussion about the Moorean phenomena, see e.g.~\cite{HollidayI10}.}). To verify that $\model,w\models\neg \tup{\tella^+_a( p)}(\Bdiam_b \top \wedge \B_b \neg  p)$, we aim for a contradiction, thus we assume  $\model,w\models \tup{\tella^+_a( p)}(\Bdiam_b \top \wedge \B_b \neg  p)$.
By $\models$, we have that: (a')  $\model,w\models \B_a p$, and that (b') $\model'',(w,e)\models (\Bdiam_b \top \wedge \B_b \neg  p)$, where $\model''= \model \otimes \tella^+_a(p)$.  
Naturally, (a') holds as a consequence of (a), thus guaranteeing the executability of the actions.

Let us analyze (b'). If $(w,e)$ has no successors via agent $b$ in $\model''$, it would contradict $\Bdiam_b\top$.
Then, since we force $(w,e)$ to have a successor via agent $b$, in such a case, $\B_b\neg p$ being true at $(w,e)$ would imply that each of its successors satisfy $\neg p$. But again, the only possible successor via agent $b$ would be one in which $f$ satisfies $p$ (by the definition of $\tella^+_a(p)$), thus we have a contradiction. Therefore,  $\model,w \models \tup{\tella^-_a(\neg p)}\B_b \neg p \rightarrow \neg \tup{\tella^+_a(p)}(\Bdiam_b \top \wedge \B_b \neg p)$. 

2) The reasoning is similar for the visual case, except that notice that $\obs{b}{\neg p}$ holds after executing $\showa^-(\neg p)$ in the current point of evaluation and does not hold after executing $\showa^+(p)$, because of the postconditions of such action types.
\end{proof}

Notice that in the validities above, all informational actions are done at a propositional level (i.e., not complex formulas). This can be seen as a limitation of the established property, but in fact it can be generalized to conjunctions of literals. Actually, we use this approach to model the French Drop trick in subsection \ref{subsec:example}. 


Then, we obtain a result we expect, based on Definitions \ref{def:sim} and \ref{def:omit}:

\begin{corollary}\label{coro:sim-dis}
For all $p\in\prop$, and for all $\model,w$, the following statement holds: 
\[
    \model,w \models (\Sim_a \neg p \wedge (\B_b\neg p \wedge \obs{b}{\neg p})) \rightarrow \Dis_a p.
\]
\end{corollary}

Corollary \ref{coro:sim-dis} states that simulation of $\neg p$ implies dissimulation of $p$, but only under the assumption that agent $b$ both believes and observes $\neg p$ (this ensures that simulation of $\neg p$ to agent $b$ actually succeeds since $b$ also believes and observes $\neg p$). More precisely, it states that when agent $a$ simulates that $\neg p$ \emph{successfully} in the sense that agent $b$ believes and observes that $\neg p$, this implies that agent $a$ also comes to dissimulate that $p$ at the same time. This implication can be understood as follows in the French Drop trick: if the magician succeeds in making the spectators wrongly believe and observe that the coin has moved from her left to right hand, then the magician has succeeded in dissimulating that the coin has remained in her left hand all along. In other words, simulating successfully that the coin has moved from left to right hand implies to dissimulate that the coin in fact remained in the magician left hand. Notice that this holds only in the case that the spectators believed and observed $\neg p$ before, otherwise dissimulation naturally may not succeed. 

The situation just described comes from the fact that we consider that the magician always succeeds in the French Drop trick (under the appropriate circumstances), so the logical behaviour is actually consistent with this reasonable assumption. Weaker actions and connections between simulation and dissimulation are worth to be explored in the future, e.g., if we consider a clumsy magician that fails in some of her actions.

%







\subsection{Magic Misdirection as a Source of Surprise}
\label{ssec:surprise}

We have seen that logic $\dol$ is strong enough to express various notions characterizing misdirection in general and in magic tricks, more specifically. But in fact, $\dol$ is also strong enough to express the sense of surprise raised by magic effects on spectators.

In the French Drop trick, it holds that $\tup{\showa^-_a(r \wedge \neg l)}(\obs{a}{l} \wedge \obs{a}{\neg r} \wedge \obs{b}{\neg l}\wedge\obs{b}{r})$, at an intermediary step. So, after a visual action of simulation $\showa^-_a(r \wedge \neg l)$, spectators observe that $\neg l$ and $r$:  $\obs{b}{\neg l}\wedge\obs{b}{r}$. Immediately after this effect, the magician unveils the mystery of the trick with visual action, i.e., it holds that $\tup{\showa^+_a(l\wedge\neg r)}(\obs{b}{l} \wedge \obs{b}{\neg r})$. In words, after the execution of $\showa^+_a(l\wedge\neg r)$, the spectators observe the true facts $l$ and $\neg r$: $\obs{b}{l}\wedge\obs{b}{\neg r}$. As one can see, those facts are completely opposite to the facts they observed just before. 

Following~\citep{Ortony&Partridge1987,Meyer&al1997,Lorini&Castelfranchi2006}, an agent is surprised by an incoming event when she recognizes an inconsistency, that is to say a \textit{discrepancy} or \textit{mismatch}, between her expectations about the world and the actual state she observes. Quite intuitively, the wider the gap between the agent's expectations and reality, the stronger the surprise. In that respect, \citep{lorini2007cognitive} distinguishes a first-order kind of surprise called ``\textit{mismatch-based surprise}'' from a second-order kind of surprise called ``\textit{astonishment}'', or ``\textit{surprise in recognition}''.

Mismatch-based surprise results from a conflict between the agent's scrutinized representation and perceived facts or events. The agent is surprised because she has some anticipatory representation, but she cannot make incoming events fit with this representation. The strength of the surprise effect depends on the agents' expectations, more precisely \textit{beliefs} about future events. Surprise occurs if an unforeseen event, in the sense of \textit{not believed} to happen, ultimately occurs. But a stronger surprise occurs if an event that was \textit{believed not} to happen ultimately happens. In the $\dol$ framework, for a given fact $p$, this can be expressed by the difference between agent $a$ not believing that $p$ when $p$ happens, i.e. $\neg \B_a p \wedge p$, versus agent $a$ believing that $\neg p$ when $p$ happens: $\B_a \neg p \wedge p$. Since the agent's disbelief that $p$ is stronger in the second case, surprise will also be stronger.   
  
Astonishment, or surprise in recognition, results from the \textit{recognition} of the absolute implausibility of a perceived fact compared to expectations. In this case, surprise is rooted in incongruity: an event not conceived as a relevant possibility actually takes place and leaves agents completely astonished. In $\dol$, agent $a$ will be astonished that a fact $p$ is true when $p$ is true but agent $a$ actually neither believes it to be true nor believes it to be false: $p \wedge \neg \B_a p \wedge \neg \B_a\neg p$. In other words, astonishment results from the fact that when $p$ occurs, $p$ is an epistemic blindspot. 

%

In the French Drop trick, the surprise the public experiences is of the mismatch-based kind. When action $\showa^-_a(r \wedge \neg l)$ leads spectators to atomic (i.e., non-epistemic) observations $\obs{b}{\neg l}$ and $\obs{b}{r}$, it also leads them to epistemically observe the same facts. The reason is that $\showa^-_a(r \wedge\neg l)$ results in $\B_a (\obs{b}{\neg l})$ and $\B_a (\obs{b}{r})$, so by definition of epistemic observation (see subsection \ref{sec:strongnotions}), we have $\Ob_b(\neg l)$ and $\Ob_b(r)$. At the end of the trick, the magician dispels her secret with visual action $\showa^+_a(l\wedge\neg r)$ now leading the spectators to atomically observe $l$ and $\neg r$: $\obs{b}{l}$ and $\obs{b}{\neg r}$. Again, spectators come to epistemically observe the same facts since action $\showa^+_a(l \wedge \neg r)$ results in $\B_a (\obs{b}{l})$ and $\B_a (\obs{b}{\neg r})$, so by definition of epistemic observation, we have: $\Ob_b(l)$ and $\Ob_b(\neg r)$. At the end of the trick, the spectators' epistemic observations $\Ob_b(l)$ and $\Ob_b(\neg r)$ are completely opposite to their epistemic observations $\Ob_b(\neg l)$ and $\Ob_b(r)$ one moment earlier, resulting in mismatch-based surprise.   

Notice that, according to \citep{lorini2007cognitive}, astonishment would have replaced mismatch-based surprise if, when revealing the trick to the audience, the magician had shown, for instance, a banknote instead of a coin, like a rabbit straight out of her hat.

\section{Conclusion}
\label{sec:final}

This article aims to conduct a logical analysis of misdirection, guided by the conceptual distinction between \textit{non-epistemic} and \textit{epistemic seeing} \cite{Dretske1970,Dretske1979}, as well as the classical contrast between \textit{simulation} and \textit{dissimulation} \cite{Bell&Whaley1991}. We adapt extant work in dynamic epistemic logic~\cite{BaltagMS98,DELbook} and in logics of attention and observation~\cite{BB23,BolanderDHLPS16} to present a dynamic logic of misdirection, denoted $\dol$. In this way we provide a novel application of (a variant of) Dynamic Epistemic Logic in the formal analysis of both verbal and visual misdirection. The static part of the logic relies on atomic formulas for representing the agents' visual observation of their environments, and on classical doxastic modality for representing their beliefs. For the dynamic part, we focus on a fragment named $\tslogic$, in which action types $\tella$ and $\showa$ are defined for expressing verbal actions of announcing information, as well as visual actions of showing information. Those actions can be either truthful (indicated by label~$^+$) or untruthful (indicated by label~$^-$), making it possible to represent actions of verbal and visual misdirection. We used $\dol$ to model a magic trick called the ``French Drop'' which illustrates actions of visual simulation (\textit{fake passing}) and of visual dissimulation (\textit{genuine palming}). The framework proved to be powerful enough to express the various notions involved in misdirection but also more complex notions, such as stronger belief attitudes, epistemic observations and the logical interactions between simulation and dissimulation. Interestingly, the setting also succeeded at capturing the specific sense of surprise resulting from magical effects.


We argue that our framework sets the basis for a broader formal theory of misdirection. Epistemic and modal logics in general are suitable to model other forms of misdirection, such as withholding information or conspiration between a set of agents. For instance, we can represent collective acts of misdirection (i.e., a magician having accomplices) by generalizing informational actions with respect to more than one agent, by using e.g. \emph{private announcements}.
On the other hand, it would be interesting to investigate different conditions on the modality of observation. In this work, we only ask for the underlying relation to be serial, in correspondence with the $\axm{D}$ axiom, while other constraints might be imposed (see e.g.~\cite{bonnay&egre2007,bonnay&egre2009} concerning non-transitive knowledge used for modeling perceptual indistinguishability). 
Moreover, the classes of models with serial relations are not closed under dynamic updates, which can also stand as an issue. Thus, we would like to explore other alternatives in the future, as well as preservation properties on the classes of models.

Another interesting notion we should integrate is the \emph{intentional} aspect of misdirection, i.e., by formalizing which kinds of intention are involved in deception. Here we assumed misdirection to be always intentional and do not investigate further on this aspect. But recent literature covers the intentional dimension from different perspectives, for instance the approach in~\cite{Sakama2021} uses Epistemic Causal Logic to that end. It would be interesting to investigate whether such a setting could accommodate our dynamic logic analysis. Other approaches exist, like~\cite{BonnetLLS2021} (using  STIT logic-based framework) and~\cite{WardTB23} (using causal games) formalizing intentional aspects of misdirection, that we should also take into account for comparison.

Finally, our approach based on dynamic epistemic logic allows us to use tools  from \emph{epistemic planning}~\cite{BolanderCPS20}. As a framework combining standard AI planning with dynamic epistemic logic, epistemic planning could help model strategies of misdirection more accurately. 
Particularly interesting is the  possibility of dealing with agents' (partial) observations. 
It would be interesting to study the aforementioned task, but where the initial state is an observational epistemic model, and where the action models and the goal(s) are defined over $\dol$, or more specifically from $\tslogic$. By doing so, we would be able to determine, given some state and a set of actions, if there is an observability-based plan, which leads from the initial state to a (potentially bogus) goal. Another logical tool for reasoning about strategies of misdirection is the family of logics of \emph{knowing how}~\cite{Wang15lori,Wang2016}, specially those connected to epistemic planning as~\cite{Liu&Wang2013}. Therein, an agent \emph{knows how to achieve a certain goal} if there exists a sequence of epistemic actions, i.e. an epistemic plan, that transforms the actual model into a model satisfying such a goal. Again, it would be possible to incorporate the action types of $\dol$, or variations of it, into this framework. All these lines of investigation have been briefly discussed, and deserve further exploration in the future.


\paragraph*{Acknowledgments.} The authors thank two anonymous reviewers for their helpful comments and suggestions, as well as Deborah Marber, Paul Égré, and François Olivier for their feedback on earlier versions of this paper. BI acknowledges the ANRs HYBRINFOX (ANR-21-ASIA-0003) and TRUSTEDNEWS
(ANR-25-ASM2-0003), and the programs ECOS-SUD (“Logical consequence and many-valued models”, no. A22H01) and THEMIS (grant agreements n°DOS022279400 and n°DOS022279500). RF is partially supported by projects ANPCyT-PICT-2021-GRF-TI-00400, Stic-AmSud 23-STIC-07 ‘DL(R)’, SecytUNC, the EU Grant Agreement 101008233 (MISSION), and the IRP SINFIN.

\bibliographystyle{plain}
\bibliography{bib}

\end{document}